\documentclass{article}

\usepackage{arxiv}
\usepackage{thumbpdf,lmodern}
\usepackage{bm}
\usepackage{tikz-qtree}
\usetikzlibrary{positioning}

\usepackage[utf8]{inputenc}
\usepackage{xcolor}
\usepackage{amsmath,amssymb,amsthm}
\usepackage{mathabx}   
\usepackage{graphicx}
\usepackage{booktabs}
\usepackage{pdflscape}
\usepackage{cancel}
\usepackage{hyperref}
\usepackage{multirow}
\usepackage{natbib}

\newcommand{\review}[1]{ {\color{red} #1 }}

\newcommand\undermat[2]{%
  \makebox[0pt][l]{$\smash{\underbrace{\phantom{%
    \begin{matrix}#2\end{matrix}}}_{\text{$#1$}}}$}#2}

\newtheorem{theorem}{Theorem}

\newtheorem{lemma}{Lemma}

\graphicspath{ {./figs/examples/}{./figs/sims/}}

\begin{document}

\title{Rediscovering Bottom-Up: Effective Forecasting in Temporal Hierarchies}

\author{
  Lukas Neubauer \\
  TU Wien\\
  \texttt{lukas.neubauer@tuwien.ac.at} \\
   \And
 Peter Filzmoser \\
  TU Wien\\
  \texttt{peter.filzmoser@tuwien.ac.at} \\
}

\maketitle

\begin{abstract}
Forecast reconciliation has become a prominent topic in recent forecasting literature, with a primary distinction made between cross-sectional and temporal hierarchies. This work focuses on temporal hierarchies, such as aggregating monthly time series data to annual data. We explore the impact of various forecast reconciliation methods on temporally aggregated ARIMA models, thereby bridging the fields of hierarchical forecast reconciliation and temporal aggregation both theoretically and experimentally. Our paper is the first to theoretically examine the effects of temporal hierarchical forecast reconciliation, demonstrating that the optimal method aligns with a bottom-up aggregation approach. To assess the practical implications and performance of the reconciled forecasts, we conduct a series of simulation studies, confirming that the findings extend to more complex as well as possibly misspecified models. This result helps explain the strong performance of the bottom-up approach observed in many prior studies. Finally, we apply our methods to real data examples, where we observe similar results.
\end{abstract}


\clearpage

\section{Introduction}

Forecast reconciliation has been a very popular topic in recent forecasting literature. It covers the questions on how to properly forecast time series which have been aggregated in a certain way. This aggregation could come from a cross-sectional aspect where a collection of time series is aggregated across different variables such as location or organizational unit. In contrast, the time series could also be aggregated on a temporal basis, such as monthly, quarterly, and annual time series. Naturally, both types of aggregation might be combined in any way, leading to cross-temporal hierarchies. 

The field of hierarchical forecast reconciliation investigates how to handle forecasting those hierarchies such that the resulting forecasts match the aggregation properties of the hierarchy. In addition, it is often examined how the performance of the reconciliation methods yielding so-called coherent forecasts is compared to original, possibly non-coherent forecasts. A very recent and extensive review of forecast reconciliation is given in \citet{ATHANASOPOULOS2024430}. Many extensions are discussed such as adding complex constraints (non-negativity, integer-based time series, ...) or probabilistic forecasting.

In this paper we investigate temporal hierarchies as introduced by \citet{ATHANASOPOULOS201760}. The authors argue that already existing forecast reconciliation methods can be applied to temporally aggregated time series in a straightforward manner. However, no further assumptions besides the base forecasts being unbiased are investigated, especially since no work is available looking at the theoretical implications of reconciliation methods assuming certain data-generating processes. We fill this gap of research and examine the performance of forecast reconciliation in temporal hierarchies in the theoretical framework of temporally aggregated time series models such as ARIMA models. 

The effects of temporal aggregation in autoregressive models were first studied by \citet{AmemiyaWu1972}. The authors prove that if some data is generated by an autoregressive model of order $p$, then a non-overlapping aggregate of these data will also follow a similar generating process. Namely, the autoregressive order of the aggregate remains at the same order $p$ while there might exist a moving average part of a certain order as well. In fact, the authors give a maximum order for this moving average part of the process. \citet{Silvestrini2005TEMPORALAO} give a generalized overview of this theory and extend it to general SARIMA models.

In temporal hierarchies, simple reconciliation techniques such as bottom-up approaches are often applied. A bottom-up forecast is generated by aggregating the forecasts of the disaggregated series. \citet{Ramrez2014InsightsIT} suggest that forecasts of aggregated time series can be improved by using bottom-up forecasts, as long as the aggregated model includes a significant moving average component. Without this component, the improvements may be minimal or nonexistent. In this work, we extend this analysis by considering more complex models and more intricate temporal hierarchies.

We take an additional step to analyze the performance of the bottom-up approach compared to more sophisticated reconciliation methods, thereby linking the fields of temporal forecast reconciliation and temporally aggregated time series models. Although this was experimentally examined in \citet{ATHANASOPOULOS201760}, the results have yet to be theoretically justified. In general, the connection between these two fields has not been established from a theoretical perspective.

The paper is structured as follows. In Section~\ref{sec:meth} we briefly discuss the ideas of hierarchical forecast reconciliation and recent advances (Section~\ref{sec:fcrecon}), in particular regarding temporal hierarchies (Section~\ref{sec:temp_fr}) as well as the basics of temporally aggregated time series models (Section~\ref{sec:tempagg}). This is followed by the linkage of the two topics in Section~\ref{sec:recon_tempagg} where we discuss the theoretical implications of forecast reconciliation on the temporally aggregated time series. In Section~\ref{sec:exps}, we investigate the discussed implications in various simulation studies, followed by real data applications in Section~\ref{sec:real_data}. Finally, we give concluding remarks in Section~\ref{sec:concl}.

\section{Related Work}\label{sec:meth}

\subsection{Hierarchical Forecast Reconciliation}\label{sec:fcrecon}

Forecast reconciliation aims to produce coherent forecasts that adhere to the aggregation structure of a time series hierarchy. The seminal work by \citet{hyndman:opt_fc_comb} introduced a regression-based approach, formulating it as a generalized least squares problem
\begin{align}\label{eq:reg}
    \hat{\mathbf y}_h = S\bm\beta_h+\bm\epsilon_h,
\end{align}
where $\hat{\mathbf y}_h$ are the base forecasts, $\bm\beta_h$ are the regression coefficients, and $\bm\epsilon_h$ is the reconciliation error with covariance matrix $V_h$. The generalized linear solution is $\hat{\bm\beta}_h=G_h\hat{\mathbf y}_h$ and reconciled forecasts $\tilde{\mathbf y}_h=SG_h\hat{\mathbf y}_h$, where $G_h = (S'V_h^{-1}S)^{-1}S'V_h^{-1}$.

\citet{wick:opt_fc_recon} proposed the minimum trace (minT) estimator based on the fact that $V_h$ is not identifiable, which minimizes the trace of the reconciled forecast error covariance matrix. Namely, 
\begin{align}\label{eq:mint}
    \min_{G} \text{tr}~\text{Cov}(\mathbf y_{T+h|h}-\tilde{\mathbf y}_h) = \min_G \text{tr}~SG W_h G'S',
\end{align}
subject to $SGS=S$, or equivalently, preserving unbiasedness of the forecasts. This leads to $G_h = (S'W_h^{-1}S)^{-1}S'W_h^{-1}$. This generalizes the regression-based solution of \eqref{eq:reg} and ensures forecast coherence as well as unbiasedness while minimizing forecast errors across all hierarchical levels.

A key challenge is estimating the covariance matrix of the base forecast errors, $W_h$, which is not always feasible, especially for complex hierarchies and long forecast horizons. To address this, researchers have proposed simpler estimators, such as equal weighting ($V_h=k_h I_n$), scaled reconciliation ($V_h=k_h\text{diag}(W_h)$) (see \citet{HYNDMAN201616}), sample and shrinkage estimators for $W_h$, and structural scaling ($W_h=k_h\text{diag}(S \mathbf 1_{n_b})$) (see \citet{wick:opt_fc_recon}).

The minimum trace method offers several advantages: it produces coherent forecasts, maintains unbiasedness if the base forecasts are unbiased, and enhances forecast performance by minimizing the overall forecast error variance. However, the potential improvements come with the challenge of accurately estimating the required covariance matrices, and \citet{PANAGIOTELIS2021343} argue that for some realizations, the performance of the reconciled forecasts may be worsened since the minT approach minimizes an expected loss function.

\subsection{Temporal Hierarchical Forecast Reconciliation}\label{sec:temp_fr}

Temporal hierarchies can be applied to forecast reconciliation as discussed by \citet{ATHANASOPOULOS201760}. This involves aggregating time series data at different levels (e.g., quarterly, biannual, annual), creating a hierarchical structure. A motivating example is given in Figure~\ref{fig:421_hier}. As previously described, the reconciliation process uses a regression-like approach, minimizing the trace of the forecast error covariance matrix. However, estimating this matrix is challenging due to limited samples at higher aggregation levels.

The fact that $\mathbf y_i = S\mathbf y_i^{[1]}$ suggests we can set up a very similar regression problem based on the base forecasts. The minimum trace approach then yields
\begin{align}
    \mathbf{\tilde y}_h = S(S'W_h^{-1}S)^{-1}S'W_h^{-1}\mathbf{\hat y}_h,
\end{align}
where $\mathbf{\hat y}_h$ are the stacked base forecasts across the entire hierarchy, and $W_h=\text{Cov}(\mathbf y_h-\mathbf{\hat y}_h)$ denotes the covariance matrix of the stacked base forecast errors.  Specifically, this means that on each aggregation level, we require $M_k h$ steps ahead forecasts, where $M_k$ denotes the frequency of aggregation level $k$. In the motivating example of Figure~\ref{fig:421_hier} this is $M_4 = 4, M_2 = 2, M_1 = 1$ and $h>0$ denotes the forecast horizon based on the top level of the hierarchy.

Researchers have proposed various simplified estimators. These include scaled reconciliation, similar to \citet{hyndman:opt_fc_comb}, and structural scaling, as in \cite{wick:opt_fc_recon}. \citet{NYSTRUP2020876} suggested autocorrelation-based methods such as autocovariance scaling, Markov scaling, GLASSO for inverse cross-correlation estimation, and cross-correlation shrinkage.
Further work by \citet{NYSTRUP20211127} explored dimension reduction techniques, using eigendecomposition of the cross-correlation matrix to create a filtered precision matrix. This approach is particularly useful for complex, deep hierarchies.
These methods aim to improve forecast accuracy by leveraging information across different temporal aggregation levels.

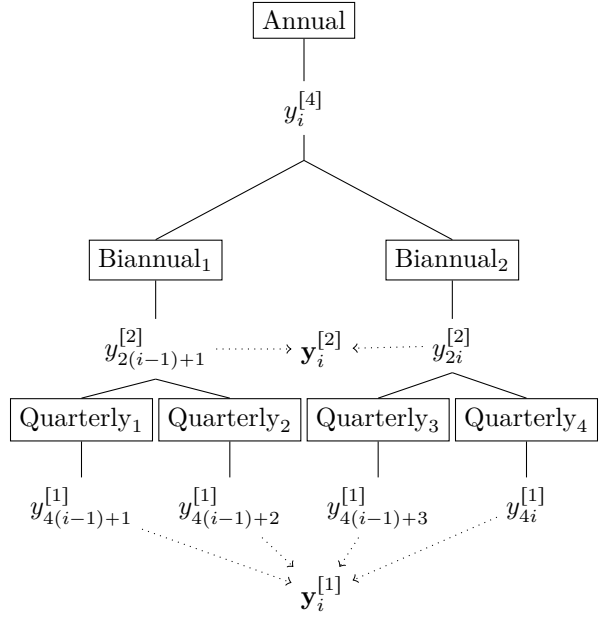
\begin{figure}[!ht]
    \centering
    \begin{tikzpicture}
        \tikzset{every tree node/.style={align=center,anchor=north}}
        \Tree [.\node[draw]{Annual}; 
            [.$y_{i}^{[4]}$ 
            [
            [.\node[draw]{$\text{Biannual}_1$};
                [.\node(bi1){$y_{2(i-1)+1}^{[2]}$};
                    [.\node[draw]{$\text{Quarterly}_1$}; \node(q1){$y_{4(i-1)+1}^{[1]}$}; ] 
                    [.\node[draw]{$\text{Quarterly}_2$}; \node(q2){$y_{4(i-1)+2}^{[1]}$}; ] 
                ] 
            ]
            [.\node[draw]{$\text{Biannual}_2$};
            [.\node(bi2){$y_{2i}^{[2]}$};
                    [.\node[draw]{$\text{Quarterly}_3$}; \node(q3){$y_{4(i-1)+3}^{[1]}$}; ] 
                    [.\node[draw]{$\text{Quarterly}_4$}; \node(q4){$y_{4i}^{[1]}$}; ] 
                ] ]
        ]]]

    \node[left = of bi2, right = of bi1] (vec2) {$\mathbf y_i^{[2]}$};
    \node[left = of bi2, right = of bi1, below = 1in of vec2] (vec1) {$\mathbf y_i^{[1]}$};

    \begin{scope}[dotted]
        \draw [->] (bi1)--(vec2);
        \draw [->] (bi2)--(vec2);
        \draw [->] (q1)--(vec1);
        \draw [->] (q2)--(vec1);
        \draw [->] (q3)--(vec1);
        \draw [->] (q4)--(vec1);
    \end{scope}
    
    \end{tikzpicture}
    \caption{Visualization of an annual-biannual-quarterly temporal hierarchy. }
    \label{fig:421_hier}
\end{figure}

\subsection{Temporal Aggregation}\label{sec:tempagg}

The temporal aggregation of time series was first explored in the work of \citet{AmemiyaWu1972}. A more recent review of key advancements can be found in \citet{Silvestrini2005TEMPORALAO}, where ARIMA-based models are primarily discussed. 

Given a univariate series $y=(y_t, t=1,2,\dots)$, a $k$-aggregated series $y^\ast$ is defined as:
\[
y_t^\ast = \sum_{i=0}^k w_{i} y_{t-i}.
\]
For non-overlapping aggregates, the new time scale is \(T=kt\), so \(y^\ast\) represents a lower-frequency series observed every $k$ steps.

The most common aggregation method is flow aggregation (with \(w_i=1\)), which sums the values over each period. Another approach is stock aggregation, where only the last observation (\(w_0=1\)) is used. Like most literature, we focus on flow aggregation.

If the higher-frequency series \(y\) follows an $\text{ARIMA}(p,d,q)$ model, then the aggregated series \(y^\ast\) is also of the ARIMA family with
\begin{align}\label{eq:agg_arima}
    y^\ast \sim \text{ARIMA}(p,d,r),\quad r\leq \left\lfloor\frac{p(k-1)+(d+1)(k-1)+q}{k}\right\rfloor.
\end{align}
The autoregressive and integrated orders remain unchanged, but the moving average order increases. The AR roots of the aggregated model are the $k$-th power of the original AR roots, so as aggregation increases, the AR effect diminishes, while the MA effect grows. Calculating MA coefficients is complex, requiring non-linear equations involving autocorrelations and other parameters.

This theory extends to more complex models like ARIMAX, SARIMA, and even volatility models like GARCH. Polynomial cancellations in disaggregated models can simplify aggregated models, as demonstrated in the case of an $\text{AR}(9)$ model reducing to $\text{AR}(3)$ after aggregation (\citet{Ramrez2014InsightsIT}).

\citet{Ramrez2014InsightsIT} also explore forecasting with aggregated series, showing that forecasting from disaggregated data in terms of a bottom-up forecast can reduce errors of the aggregated data if the MA part is significant, though the improvement is minimal when it is not.

A detailed calculation for an $\text{AR}(1)$ model is provided in Appendix \ref{app:calc_proofs}.

\section{Temporal Hierarchical Forecast Reconciliation in Temporally Aggregated Models}\label{sec:recon_tempagg}

In this section, we will theoretically integrate the fields of temporal forecast reconciliation and temporally aggregated ARIMA models. To the best of our knowledge, this is the first time such an integration has been attempted. While \citet{ATHANASOPOULOS201760} utilized the theory of temporally aggregated ARIMA models, their approach was primarily experimental. They examined the performance of temporal forecast reconciliation methods, such as variance scaling, and compared them to a simple bottom-up approach under varying levels of uncertainty. Specifically, they conducted experiments with fixed model orders and parameters, fixed orders alone, or automatically selected models based on model selection criteria. The authors found that temporal forecast reconciliation and bottom-up methods perform equally well in highly certain settings, but the performance of bottom-up methods declines when models are misspecified.

In general, the data-generating process has not been of much interest so far in the field of temporal forecast reconciliation because it has been developed as a post-hoc procedure to transform base forecasts coherently and later to improve possibly suboptimal base forecasts originating from misspecified models. In the theory of temporally aggregated models, the combination of forecasts of different levels to achieve coherent or even better forecasts has not been looked at.

Our contribution is as follows. Utilizing the theoretical model of aggregation, we will derive the theoretical covariance matrix of the base forecast errors, denoted as $W$, given in Lemma~\ref{lm:agg_cov}. This covariance matrix will then be employed to perform the minimum trace estimation manually. Through matrix algebra, we will demonstrate in Theorem~\ref{thm:mint_bu} that the resulting mapping matrix $G$ corresponds to a bottom-up forecast. Consequently, we show that within the framework of aggregated ARIMA models, the optimal forecast reconciliation technique is indeed the bottom-up approach.

Building on the insights from Section~\ref{sec:tempagg}, we aim to manually implement the minimum trace reconciliation method. To do this, we need the covariance matrix of the base forecast errors, which we can readily compute. To maintain simplicity, we will initially focus on the straightforward case of an $\text{AR}(1)$ model and subsequently discuss more complex models. The first result in Lemma~\ref{lm:agg_cov} is about the covariance structure of the aggregated model. Its proof can be found in \ref{app:calc_proofs}.

\begin{lemma}\label{lm:agg_cov}
    The covariance matrix $W_1$
    of $1$-step forecast errors in a $k$-aggregated $\text{AR}(1)$ model with parameter $\phi$ and innovation variance $\sigma^2$
    is equal to
    \begin{align}
        W_1 &= \begin{pmatrix}
            \sigma_{\ast}^2 & \sigma^2 \mathbf 1_k'\Phi\Phi' \\
            \sigma^2 \Phi\Phi'\mathbf 1_k & \sigma^2 \Phi\Phi'
        \end{pmatrix}
    \end{align}
    where $\mathbf 1_k$ denotes a vector of ones of length $k$, $\Phi$ is a lower triangle matrix given by
    \begin{align}
        \Phi &= \begin{pmatrix}
            1 & 0 & 0 &\dots & 0 \\
            \phi & 1 & \ddots & \ddots & \vdots\\
            \vdots & \ddots & \ddots & \ddots & \vdots\\
            \phi^{k-2} & \ddots & \ddots & \ddots & 0 \\
            \phi^{k-1} & \phi^{k-2} & \dots &\phi & 1
        \end{pmatrix},
        \end{align}
    and $\sigma_{\ast}^2$ denotes the innovation variance of the aggregated model.
\end{lemma}

Based on Lemma~\ref{lm:agg_cov} we now manually compute the optimal unbiased reconciliation matrix, summarised in Theorem~\ref{thm:mint_bu}. The proof is available in \ref{app:calc_proofs}.

\begin{theorem}\label{thm:mint_bu}
    The minimum trace reconciliation method in a $k$-aggregated $\text{AR}(1)$ model is equal to a bottom-up approach, implying that
    \begin{align*}
        SG^\ast = \begin{pmatrix}
            0 & \mathbf 1_k' \\
            \mathbf 0_k & \mathbf I_k
        \end{pmatrix},
    \end{align*}
    where $\mathbf 0_k$ is a vector of zeros of length $k$, $\mathbf I_k$ is the $k\times k$ identity matrix, and $G^\ast$ denotes the optimal mapping matrix from problem (\ref{eq:mint}).
\end{theorem}

Theorem~\ref{thm:mint_bu} indicates that the optimal unbiased reconciliation method for the aggregated $\text{AR}(1)$ model is the bottom-up approach. Consequently, the forecasts at the bottom level remain unchanged, with no potential for enhancing forecast accuracy. Conversely, the aggregated forecast is disregarded in any form of combination. This outcome elucidates why the bottom-up approach frequently demonstrates effectiveness in both simulation studies and real-world data applications, thus bolstering its practicality.

Before moving on to the experimental part of this study, we aim to illustrate how this theorem works using a sample-based approach. In Figure~\ref{fig:ar1_SG_2}, the average transformation matrix $SG$ for a two-level hierarchy is presented. To do this, we simulated 100 models and estimated the complete sample covariance matrix based on the simulations. The models used consist of an $\text{AR}(1)$ model with parameters $\phi=0.8,\sigma^2=1$ at the lower level, which is then combined into an $\text{ARMA}(1,1)$ model at the higher level of the hierarchy with $k\in\{4,1\}$. The nodes of the hierarchy are shown on both axes, with $1-1$ representing the entry at the top level and $2-i$ representing the $i$-th step of the lower level. This precisely specifies the transformation matrix as used in Theorem~\ref{thm:mint_bu}. The first row shows the effects of the base forecasts on the reconciled top level forecast. It is evident that there is little impact from the top level base forecast, with nearly equal weights close to 1 for the bottom level base forecasts. Similarly, the following $4$ rows demonstrate the weights for the reconciled bottom level forecasts, with a zero column followed by the identity matrix $I_4$. This indicates that the reconciled bottom level forecasts closely match the bottom level base forecasts. In summary, the tendency for a bottom-up reconciliation approach is clear.

\begin{figure}[!ht]
    \includegraphics[width=\textwidth]{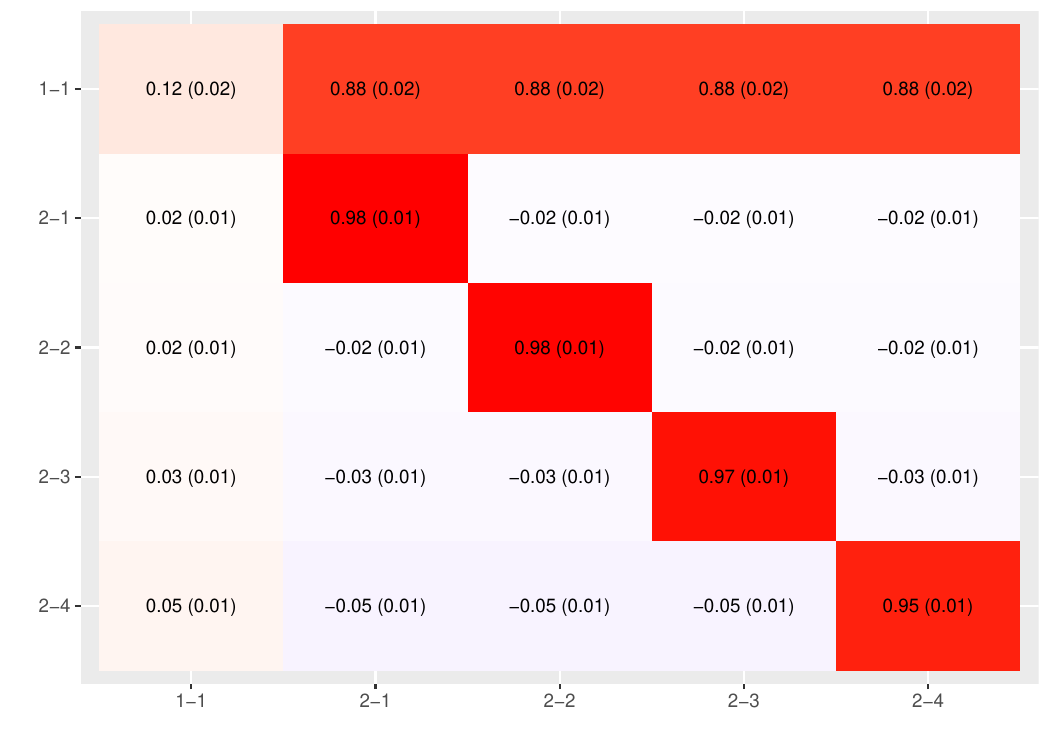}
    \caption{Full sample transformation matrix $SG$ for $n=100,\phi=0.8, h=1, k\in\{4,1\}, \sigma^2=1$. The colors correspond to the mean value over $100$ repetitions. The standard errors are given in parentheses. }
    \label{fig:ar1_SG_2}
\end{figure}

\section{Experiments}\label{sec:exps}

In this section, we experimentally investigate different types of forecast reconciliation methods in the framework of temporally aggregated time series models and beyond.

We evaluate the results based on percentage errors, namely for aggregation parameter $k$ we obtain a relative mean squared error of
\begin{align*}
    \text{rMSE}^{[k]}(\tilde{\mathbf y}, \hat{\mathbf y}) = \frac{\sum_i \left\|\tilde{\mathbf y}_i^{[k]} - \mathbf y_i^{[k]}\right\|_2^2}{\sum_i \left\|\hat{\mathbf y}_i^{[k]} - \mathbf y_i^{[k]}\right\|_2^2} - 1,
\end{align*}
where $\tilde{\mathbf y}_i^{[k]}$ denotes the $i$-th vector of reconciled forecasts of aggregation level $k$,
$\hat{\mathbf y}_i^{[k]}$ is the $i$-th vector of the base forecasts of aggregation level $k$, and $\|\cdot\|_2^2$ is the squared Euclidean norm. We analyze both in-sample (training) reconciliation errors and out-of-sample (test) reconciliation errors to assess generalizability, aggregating the corresponding observations accordingly. Depending on the level of aggregation, we may encounter multi-step ahead forecasts. To simplify, we aggregate these multi-step forecasts, providing a single error measure for each aggregation level.

The test reconciliation forecasts are acquired through the following procedure. The reconciliation method employed is trained exclusively on the training data, meaning that the covariance matrix and the corresponding base ARIMA models are estimated solely based on the training data. Subsequently, forecasts for $M_k h$ steps ahead are generated for the test data in a cumulative manner, effectively utilizing the test data for the base test forecasts. The base ARIMA models are modelled by the \texttt{Arima} function whereas automated model selection is based on the \texttt{auto.arima} function (based on the corrected AIC measure, developed by \citet{aicc}), both from the \texttt{forecast} package in R (\citet{forecast-pkg}).

MSE values are computed for each level of the hierarchy as well on an overall level by taking the sum of MSEs across all levels. The reason we consider MSE instead of a different error measure is that the minimum trace reconciliation method exactly minimizes the sum of the error variances.

For a robustness check of the results, we also consider a relative mean absolute error and use it to calculate percentage errors. Namely,
\begin{align}
     \text{rMAE}^{[k]}(\tilde{\mathbf y}, \hat{\mathbf y}) = \frac{\sum_i \left\|\tilde{\mathbf y}_i^{[k]} - \mathbf y_i^{[k]}\right\|_1}{\sum_i \left\|\hat{\mathbf y}_i^{[k]} - \mathbf y_i^{[k]}\right\|_1} - 1,
\end{align}
where $\|\cdot\|_1$ is the absolute-value norm. This error measure is inherently less sensitive to outliers. We have focused on reporting results for rMSE to keep things concise. The conclusions remain consistent even when considering rMAE or similar relative error measures.

Overall, if a percentage error is below $0$, this indicates that the reconciled forecasts perform better, whereas errors above $0$ suggest the opposite. It is important to note that we are only examining relative errors, focusing on the performance of the temporally reconciled forecasts rather than the base forecasts. Our aim is to evaluate how different types of temporal forecast reconciliation methods perform.

For each of the following settings we simulate $N=100$ time series and compute training and test rMSE values. The training data always consist of $75\%$ of the total data.

The covariance estimators we focus on in the simulations are
\begin{itemize}
    \item OLS: $\hat W_h = k_h I$,
    \item Full Cov.: $\hat W_h = \frac{1}{0.75n} \sum_{i=1}^{0.75n} \left(\hat{\mathbf e}_{i}^{(h)}\right)\left(\hat{\mathbf e}_{i}^{(h)}\right)'$, where $\hat{\mathbf e}_{i}^{(h)}$ denote the $i$-th vector of $h$-step residuals of the base forecasts, and
    \item Spectral Scaling \citep{NYSTRUP20211127}: 
    \begin{enumerate}
        \item Shrink the empirical cross-correlation matrix $R$ to $R_\text{shrink}=(1-\nu)R+\nu I$
        \item Eigen-decompose this shrunk cross-correlation matrix by $R_\text{shrink}=V\Lambda_\text{shrink}V'$ where $R=V\Lambda V'$. 
        \item Reconstruct the filtered precision matrix by $Q=(WAW'+cI)^{-1}$ such that $W$ contains the first $n_\text{eig}$ columns of $V$ and $A=\text{diag}((1-\nu)\lambda_1+\nu - c, \dots, (1-\nu)\lambda_\text{neig}+\nu - c)$ with $c$ being the average of the remaining smallest shrunken eigenvalues.
        \item Set $\hat W_h^{-1} = D_{\text{var}}^{-1/2}QD_{\text{var}}^{-1/2}$ where $D_{\text{var}}$ corresponds to variance scaling.
    \end{enumerate}
    The two hyperparameters $\nu,n_\text{eig}$ are chosen in a time series cross-validation procedure as discussed in \cite{BERGMEIR2012192_tscv}. We should not use regular cross-validation because of non-independent data. The authors do not follow this procedure and rather rely on an optimally chosen shrinkage parameter $\nu$ (\cite{LEDOIT_WOLF_SHRINKAGE}) and a fixed number of chosen eigenvectors $n_\text{eig}$.
\end{itemize}

Other estimators, including various shrinkage estimators and scaling variants, were initially considered in the simulations but produced results very similar to those listed. Additionally, the bottom-up approach was also examined. 

\subsection{Autoregressive Models of Order 1}
In the first experiment, we want to demonstrate the implications of Theorem~\ref{thm:mint_bu}. We simulate stationary $\text{AR}(1)$ data on the bottom level of the hierarchy and aggregate them to obtain the remaining levels of the hierarchy. The parameters we vary are
\begin{itemize}
    \item Sample size on the top level $n=20,50,100$,
    \item AR parameter $\phi=-0.9,\dots,0.9$,
    \item Hierarchy size $k\in\{4,1\},\{5,1\},\{12,4,1\}$,
    \item Forecast horizon based on the top level $h=1,2$, and
    \item Fixed order of the ARMA models to remove model uncertainty which corresponds to Scenario $2$ of \citet{ATHANASOPOULOS201760}, or automated model selection (Scenario $3$).
\end{itemize}
The innovation variance on the bottom level $\sigma^2=1$ is fixed. In order to save space, we only report results for $h=1$ and the simple hierarchy $k\in\{4,1\}$. The remaining settings lead to similar conclusions.

Table~\ref{tab:h=1, k=(1,4), sigma.sq=1, auto=FALSE} presents the training and test rMSE values for the selected reconciliation methods and parameters, grouped by buckets of the AR parameter. This allows us to distinguish between high negative or positive correlation as well as almost random walks. We observe that most improvements occur at the top level of the hierarchy, while reconciliation at the bottom level yields worse results, especially out-of-sample. Overall, we notice similar improvements for the bottom-up approach compared to more sophisticated methods once the sample size is sufficiently large. Note that the highest improvements are observed for a large AR parameter across all methods due to the increased uncertainty in the forecasts. Table~\ref{tab:h=1, k=(1,4), sigma.sq=1, auto=TRUE} in the Appendix shows the results for auto-selected models. The increase in model uncertainty which usually results in simpler models is presented through larger standard errors. Besides the fact that the full covariance estimator is singular in many cases, the main conclusions remain the same. The possible misspecification of the base models does not impact the fact that the bottom-up approach is very suited for such setting.

\begin{table}
\centering
\caption{\label{tab:h=1, k=(1,4), sigma.sq=1, auto=FALSE}Mean rMSE per buckets of $\phi$ for $ h=1, k\in\{4,1\}, \sigma^2=1 $ and fixed order of the used models. The standard errors are given in parentheses.}
\centering
\resizebox{\ifdim\width>\linewidth\linewidth\else\width\fi}{!}{
\begin{tabular}[t]{lllrrrrrr}
\toprule
\multicolumn{3}{c}{ } & \multicolumn{3}{c}{Training rMSE} & \multicolumn{3}{c}{Test rMSE} \\
\cmidrule(l{3pt}r{3pt}){4-6} \cmidrule(l{3pt}r{3pt}){7-9}
Level & n & Recon. Type & {}[-0.9,-0.5] & (-0.5,0.5] & (0.5,0.9] & {}[-0.9,-0.5] & (-0.5,0.5] & (0.5,0.9]\\
\midrule
 &  & Bottom-Up & 0.15 (0.02) & 0.16 (0.01) & 0.03 (0.02) & \textbf{-0.08 (0.03)} & \textbf{-0.04 (0.02)} & \textbf{-0.11 (0.03)}\\

 &  & Full Cov. & \textbf{-0.03 (0.01)} & 0.01 (0.01) & \textbf{-0.10 (0.01)} & 0.11 (0.04) & 0.09 (0.02) & -0.06 (0.03)\\

 &  & Spectral & -0.03 (0.00) & 0.01 (0.00) & -0.08 (0.01) & -0.02 (0.02) & -0.01 (0.01) & -0.04 (0.04)\\

 & \multirow{-4}{*}{ 20} & OLS & -0.01 (0.00) & \textbf{ 0.01 (0.00)} & -0.04 (0.00) & -0.06 (0.01) & -0.04 (0.00) & -0.08 (0.01)\\
\cmidrule{2-9}
 &  & Bottom-Up & 0.02 (0.01) & 0.06 (0.00) & -0.08 (0.01) & \textbf{-0.11 (0.01)} & \textbf{-0.07 (0.01)} & \textbf{-0.14 (0.01)}\\

 &  & Full Cov. & \textbf{-0.05 (0.00)} & \textbf{-0.01 (0.00)} & \textbf{-0.12 (0.01)} & -0.05 (0.01) & 0.01 (0.01) & -0.13 (0.01)\\

 &  & Spectral & -0.03 (0.00) & 0.00 (0.00) & -0.11 (0.01) & -0.05 (0.01) & -0.02 (0.01) & -0.13 (0.01)\\

 & \multirow{-4}{*}{ 50} & OLS & -0.02 (0.00) & 0.00 (0.00) & -0.05 (0.00) & -0.04 (0.00) & -0.03 (0.00) & -0.06 (0.00)\\
\cmidrule{2-9}
 &  & Bottom-Up & -0.03 (0.00) & 0.01 (0.00) & -0.11 (0.01) & \textbf{-0.09 (0.01)} & \textbf{-0.04 (0.00)} & \textbf{-0.17 (0.01)}\\

 &  & Full Cov. & \textbf{-0.05 (0.00)} & \textbf{-0.01 (0.00)} & \textbf{-0.13 (0.01)} & -0.06 (0.01) & 0.00 (0.00) & -0.15 (0.01)\\

 &  & Spectral & -0.04 (0.00) & -0.01 (0.00) & -0.12 (0.01) & -0.05 (0.01) & -0.02 (0.00) & -0.15 (0.01)\\

\multirow{-12}{*}[1\dimexpr\aboverulesep+\belowrulesep+\cmidrulewidth]{ Level 1} & \multirow{-4}{*}{ 100} & OLS & -0.02 (0.00) & 0.00 (0.00) & -0.05 (0.00) & -0.04 (0.00) & -0.01 (0.00) & -0.06 (0.00)\\
\cmidrule{1-9}
 &  & Bottom-Up & 0.00 (0.00) & 0.00 (0.00) & 0.00 (0.00) & \textbf{ 0.00 (0.00)} & \textbf{ 0.00 (0.00)} & \textbf{ 0.00 (0.00)}\\

 &  & Full Cov. & \textbf{-0.06 (0.01)} & \textbf{-0.06 (0.01)} & \textbf{-0.08 (0.01)} & 0.17 (0.02) & 0.24 (0.03) & 0.17 (0.03)\\

 &  & Spectral & -0.04 (0.00) & -0.04 (0.00) & -0.05 (0.01) & 0.07 (0.01) & 0.07 (0.01) & 0.12 (0.02)\\

 & \multirow{-4}{*}{ 20} & OLS & -0.01 (0.00) & -0.03 (0.00) & 0.00 (0.01) & 0.02 (0.00) & 0.04 (0.01) & 0.16 (0.03)\\
\cmidrule{2-9}
 &  & Bottom-Up & 0.00 (0.00) & 0.00 (0.00) & 0.00 (0.00) & \textbf{ 0.00 (0.00)} & \textbf{ 0.00 (0.00)} & \textbf{ 0.00 (0.00)}\\

 &  & Full Cov. & \textbf{-0.02 (0.00)} & \textbf{-0.03 (0.00)} & \textbf{-0.03 (0.00)} & 0.05 (0.01) & 0.09 (0.01) & 0.04 (0.01)\\

 &  & Spectral & -0.01 (0.00) & -0.02 (0.00) & -0.01 (0.00) & 0.02 (0.01) & 0.05 (0.01) & 0.03 (0.01)\\

 & \multirow{-4}{*}{ 50} & OLS & 0.00 (0.00) & -0.01 (0.00) & 0.04 (0.01) & 0.01 (0.00) & 0.03 (0.01) & 0.11 (0.01)\\
\cmidrule{2-9}
 &  & Bottom-Up & 0.00 (0.00) & 0.00 (0.00) & 0.00 (0.00) & \textbf{ 0.00 (0.00)} & \textbf{ 0.00 (0.00)} & \textbf{ 0.00 (0.00)}\\

 &  & Full Cov. & \textbf{-0.01 (0.00)} & \textbf{-0.02 (0.00)} & \textbf{-0.02 (0.00)} & 0.02 (0.00) & 0.02 (0.00) & 0.02 (0.01)\\

 &  & Spectral & -0.01 (0.00) & -0.01 (0.00) & 0.00 (0.00) & 0.01 (0.00) & 0.01 (0.00) & 0.02 (0.01)\\

\multirow{-12}{*}[1\dimexpr\aboverulesep+\belowrulesep+\cmidrulewidth]{ Level 2} & \multirow{-4}{*}{ 100} & OLS & 0.00 (0.00) & 0.00 (0.00) & 0.05 (0.00) & 0.01 (0.00) & 0.01 (0.00) & 0.10 (0.01)\\
\cmidrule{1-9}
 &  & Bottom-Up & 0.07 (0.01) & 0.12 (0.01) & 0.02 (0.01) & \textbf{-0.08 (0.01)} & \textbf{-0.05 (0.01)} & \textbf{-0.12 (0.03)}\\

 &  & Full Cov. & \textbf{-0.04 (0.01)} & \textbf{-0.01 (0.01)} & \textbf{-0.11 (0.01)} & 0.10 (0.02) & 0.10 (0.01) & -0.05 (0.02)\\

 &  & Spectral & -0.03 (0.00) & -0.01 (0.00) & -0.08 (0.01) & 0.01 (0.01) & 0.00 (0.01) & -0.04 (0.04)\\

 & \multirow{-4}{*}{ 20} & OLS & -0.01 (0.00) & 0.00 (0.00) & -0.04 (0.00) & -0.04 (0.00) & -0.03 (0.00) & -0.07 (0.00)\\
\cmidrule{2-9}
 &  & Bottom-Up & 0.01 (0.00) & 0.04 (0.00) & -0.08 (0.01) & \textbf{-0.07 (0.01)} & \textbf{-0.06 (0.01)} & \textbf{-0.14 (0.01)}\\

 &  & Full Cov. & \textbf{-0.04 (0.00)} & \textbf{-0.01 (0.00)} & \textbf{-0.11 (0.01)} & -0.01 (0.01) & 0.02 (0.00) & -0.12 (0.01)\\

 &  & Spectral & -0.02 (0.00) & 0.00 (0.00) & -0.10 (0.01) & -0.02 (0.00) & -0.01 (0.01) & -0.12 (0.01)\\

 & \multirow{-4}{*}{ 50} & OLS & -0.01 (0.00) & 0.00 (0.00) & -0.04 (0.00) & -0.02 (0.00) & -0.02 (0.00) & -0.05 (0.00)\\
\cmidrule{2-9}
 &  & Bottom-Up & -0.02 (0.00) & 0.01 (0.00) & -0.10 (0.01) & \textbf{-0.06 (0.01)} & \textbf{-0.03 (0.00)} & \textbf{-0.15 (0.01)}\\

 &  & Full Cov. & \textbf{-0.03 (0.00)} & \textbf{-0.01 (0.00)} & \textbf{-0.12 (0.01)} & -0.03 (0.00) & 0.00 (0.00) & -0.14 (0.01)\\

 &  & Spectral & -0.03 (0.00) & -0.01 (0.00) & -0.11 (0.00) & -0.02 (0.00) & -0.01 (0.00) & -0.14 (0.01)\\

\multirow{-12}{*}[1\dimexpr\aboverulesep+\belowrulesep+\cmidrulewidth]{ Overall} & \multirow{-4}{*}{ 100} & OLS & -0.01 (0.00) & 0.00 (0.00) & -0.04 (0.00) & -0.02 (0.00) & -0.01 (0.00) & -0.05 (0.00)\\
\bottomrule
\end{tabular}}
\end{table}

\subsection{ARMA Models of Higher Order}
For more complex models such as $\text{ARMA}(2,2)$ and its aggregates, computing the covariance matrices of the forecast errors becomes very tedious. Therefore, we focus on experimental evaluation for these cases to investigate if the implications of Theorem~\ref{thm:mint_bu} still hold.

As the complexity of an ARMA model increases, identifying the parameter space that yields stationary models becomes non-trivial. It is particularly challenging to define stationary parameter combinations for $p,q > 2$. To address this, we randomly draw stationary parameters using the partial correlation function as described by \citet{JONES1987}.

For each combination of $p \in\{1,2\}$ and $q \in \{0,1,2\}$, we randomly draw $100$ sets of parameters ${\phi_1, \dots, \phi_p, \theta_1, \dots, \theta_q}$. To mitigate the randomness of each realization, we further simulate $20$ time series for each of the $100$ random parameter sets.

\begin{figure}[!ht]
    \includegraphics[width=\textwidth]{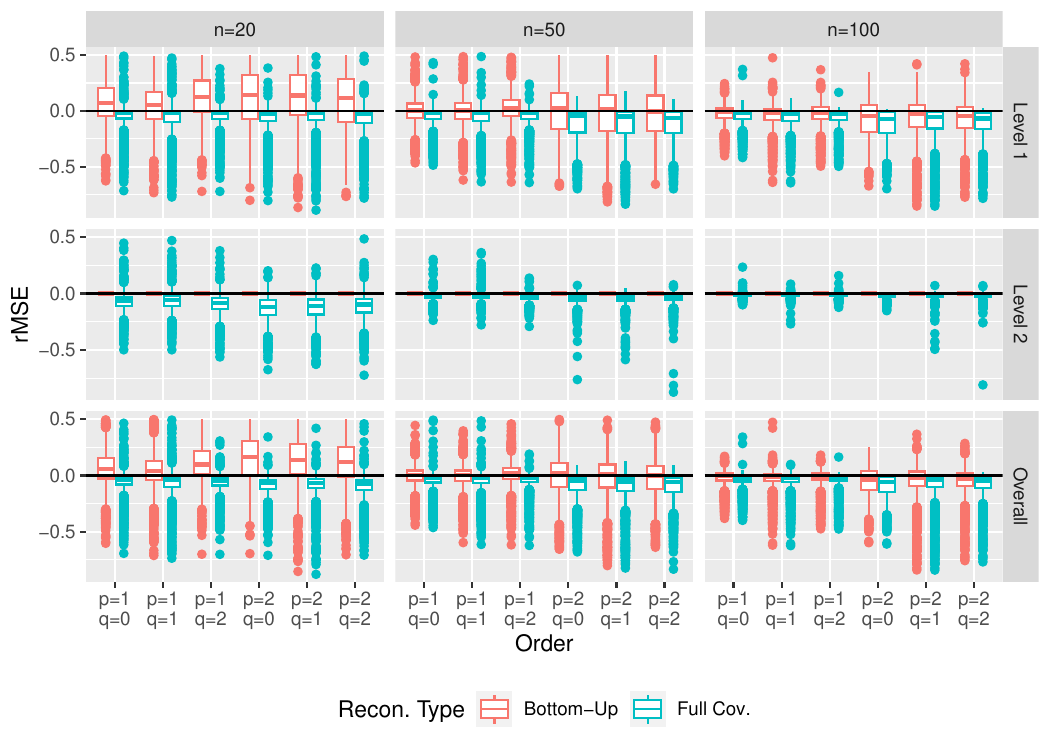}
    \caption{In-sample rMSE for various ARMA models and $h=1,k\in\{4,1\},\sigma^2=1$ and fixed-order models.}
    \label{fig:arma_1_14_1_train}
\end{figure}

Figure~\ref{fig:arma_1_14_1_train} shows the in-sample rMSE values for the full covariance estimator as well as the bottom-up approach for various sample sizes of the top level. The setting is $h=1,k\in\{4,1\}$ and $\sigma^2=1$ as well as fixed-order models. As in the $\text{AR}(1)$ for varying AR parameters, we observe equivalent reconciliation performance for a larger sample size for any $\text{ARMA}(p,q)$ present. In the low sample size case we see that bottom-up performs worse with increasing model complexity. Interestingly, this difference becomes larger for higher model complexity. We also observe that the full covariance method can produce better forecasts on the bottom level. This improvement also increases with the complexity of the bottom level base model. Overall, the MA order does not seem as impactful as the AR order.

\begin{figure}[!ht]
    \includegraphics[width=\textwidth]{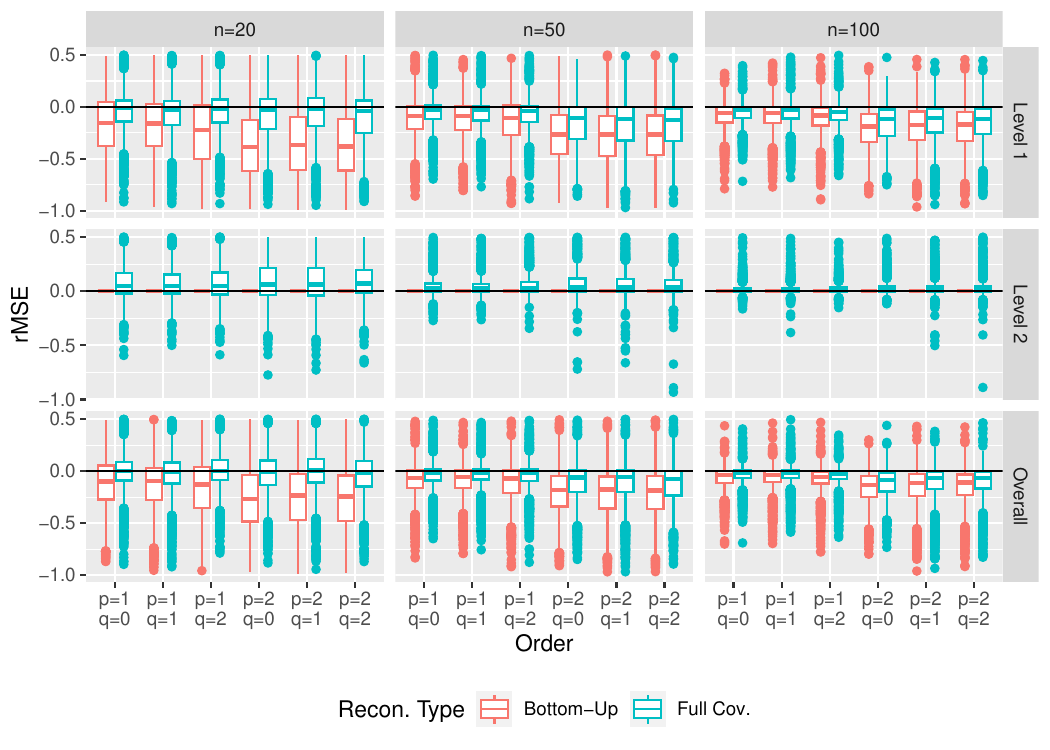}
    \caption{Out-of-sample rMSE for various ARMA models and $h=1,k\in\{4,1\},\sigma^2=1$ and fixed-order models.}
    \label{fig:arma_1_14_1_test}
\end{figure}

Figure~\ref{fig:arma_1_14_1_test} shows the test errors for the very same setting. As in the simple $\text{AR}(1)$ case, the roles of bottom-up and using the full covariance matrix estimator switch and the bottom-up approach performs better the more complex the base bottom model is set up to be.

In this analysis, we aggregate over the whole space of stationary models of a certain order. Hence we also take a look at the performance of $2$-dimensional base models in a more detailed manner. Figure~\ref{fig:ar2_1_14_1_train} shows the mean training rMSE differences between the full covariance-based reconciliation and the bottom-up approach for the randomly drawn stationary $\text{AR}(2)$ models. Based on this plot, there is no tendency for performance based on the space of the stationary parameters and we can argue that Theorem~\ref{thm:mint_bu} holds uniformly for ARMA models. Test errors are available in Figure~\ref{fig:ar2_1_14_1_test}. Similarly, Figure~\ref{fig:arma11_1_14_1_train} shows the training mean rMSE differences for $\text{ARMA}(1,1)$ models. The corresponding test errors are available in Figure~\ref{fig:arma11_1_14_1_test}. For readability, Figures~\ref{fig:ar2_1_14_1_train}--\ref{fig:arma11_1_14_1_test} are put in the Appendix.

\subsection{Misspecification of the Base Models} 

\begin{figure}[!ht]
    \includegraphics[width=\textwidth]{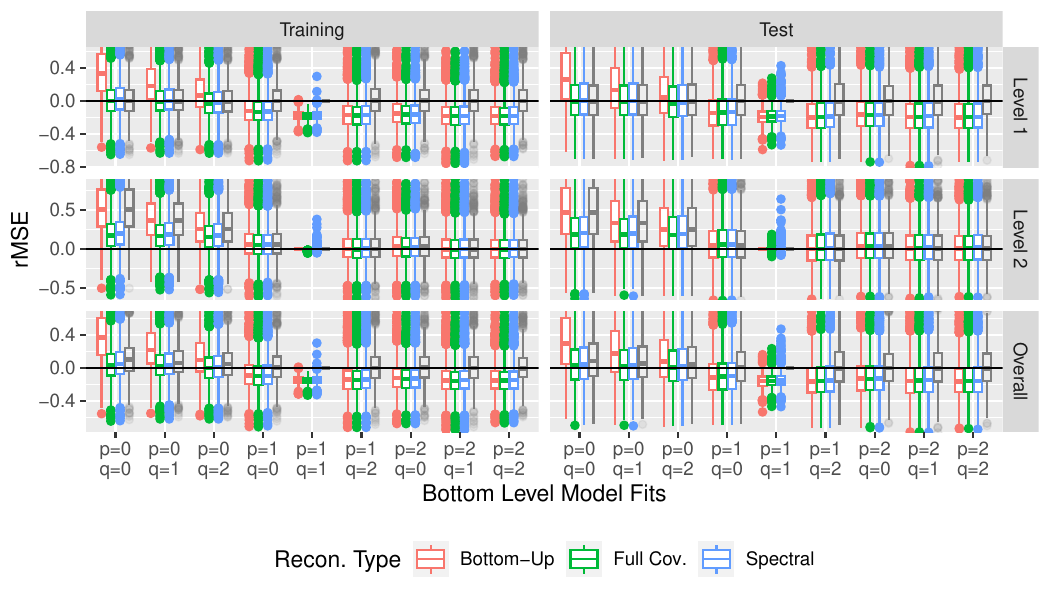}
    \caption{rMSE for various ARMA models and $h=1,k\in\{4,1\},\sigma^2=1$ in partly misspecified scenarios. The greyish boxplots indicate the relative error of the corresponding base models to the correct base model.}
    \label{fig:arma11_misspec}
\end{figure}

Next, we conduct a simulation study to better understand the impact of model misspecification. Our setup involves again the simple hierarchy with $k\in\{4,1\}, n = 100$, and $\sigma^2 = 1$. We generate bottom level data from an $\text{ARMA}(1,1)$ process, which also yields an $\text{ARMA}(1,1)$ process at the top level.

To examine how misspecification affects the bottom-up approach, we intentionally misspecify models at the bottom level. Figure~\ref{fig:arma11_misspec} displays the relative MSE values for fitted bottom level models with orders $p, q \in \{0, 1, 2\}$. This allows us to explore both under and over-fitting scenarios. We calculate the relative errors in comparison to the correctly specified model ($p = q = 1$). The errors for misspecified base models are shown as grey boxplots in Figure~\ref{fig:arma11_misspec}.

Here, we keep the top level model correctly specified. For each simulation, we generate $20$ random time series at the bottom level, each with $|\phi|,|\theta|>0.7$, and create $100$ realizations of each series. This constraint leads to clearer results, though similar patterns emerge without it.

Key findings for both training and test forecasts are as follows.
\begin{enumerate}
    \item Underfitting the autoregressive component significantly degrades bottom-up performance.
    \item Misspecification of the moving average component has minimal impact.
    \item When overfitting the autoregressive component, all forecast reconciliation methods perform similarly well. This is likely because the overfitted base models are not truly overfitting, as unnecessary parameters are estimated close to $0$.
\end{enumerate}

These results underscore the value of advanced forecast reconciliation methods compared to simpler approaches like bottom-up. Such methods can help mitigate the effects of model misspecification and resulting forecast errors, which are inevitable in real-world applications.

To ensure a comprehensive analysis, we also examine scenarios where the base models at the top level are misspecified, as illustrated in Figure~\ref{fig:arma11_misspec2}. This figure displays the test relative errors using the previously described methodology. For clarity, we only show the results of the bottom-up and the full covariance method.
Our findings indicate that misspecification at the top level, both in isolation and in combination with bottom level misspecification, has a notably less significant impact. This is evidenced by the relatively consistent boxplots across each column.
Interestingly, the bottom-up approach and the covariance-based reconciliation method demonstrate comparable performance, regardless of the degree of misspecification at the bottom level. This observation further underscores the effectiveness and reliability of the bottom-up approach in various scenarios.

\begin{figure}[!ht]
    \includegraphics[width=\textwidth]{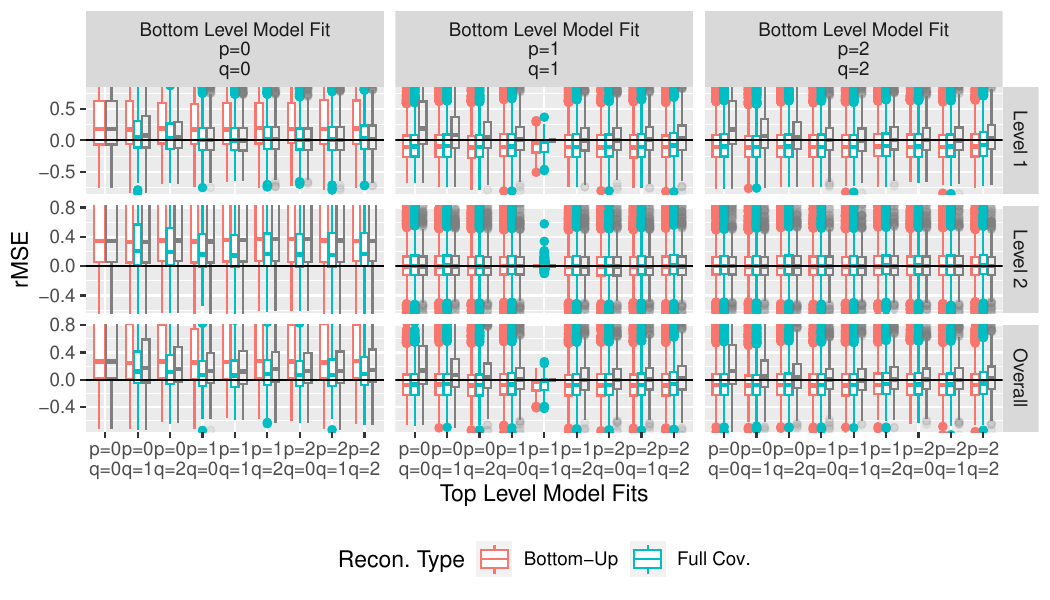}
    \caption{Test rMSE for various ARMA models and $h=1,k\in\{4,1\},\sigma^2=1$ in fully misspecified scenarios. The greyish boxplots indicate the relative error of the corresponding base models to the correct base model.}
    \label{fig:arma11_misspec2}
\end{figure}

\section{Real Data Applications}\label{sec:real_data}
\subsection{A\&E Emergency Service Demand}
Following the data example of \citet{ATHANASOPOULOS201760}, we illustrate this paper's work on the \textit{Accident \& Emergency Service Demand} dataset, available from the \texttt{thief} package in R (\citet{thief-pkg}). In this dataset, a number of demand statistics of A\&E departments are recorded on a weekly basis from $2010-11-07$ to $2015-06-07$.

Before any modeling, we perform some preprocessing. To ensure complete observations for the hierarchy, we remove the incomplete years $2011$ and $2015$, resulting in $208$ weeks of data. Next, we decompose the weekly time series of interest into seasonal, trend, and remaining components using the \texttt{stl} function in R, and remove the seasonal component. 

We analyze the \textit{Total Attendances} time series and aggregate it on a monthly basis, resulting in a small hierarchy with $52$ months of data. The training data consists of the first $41$ months, or $164$ weeks, with the remaining data designated as test data. As before, we are focused on cumulative one-step-ahead forecasts at the top level of the hierarchy, which in this case would be month-by-month forecasts. Using automated model selection and ignoring the aggregated ARIMA theory, the chosen models are $\text{ARIMA}(0,0,0)$ and $\text{ARIMA}(1,1,1)$, respectively.


To stick to the framework of temporally aggregated ARIMA models, we fix the orders of the used models accordingly. This yields an $\text{ARIMA}(1,1,2)$ model for the monthly time series. The resulting model on the top level gives an AICc value of $406.47$ which is only around $0.6\%$ worse than the automatically selected model, hence it still seems like an appropriate model. Table~\ref{tab:data:ta_wool2} shows the corresponding errors. We observe better generability of the bottom-up approach compared to using the full covariance matrix. The spectral method does seem to perform quite well out-of-sample leading to similar results as the bottom-up approach. A common aspect is still the fact that each covariance-based reconciliation method achieves worse forecasts on the test set for the bottom level time series.

\begin{table}
    \centering
    \caption{\label{tab:data:ta_wool2}Results for A\&E Total Addendances in units of $\text{People}^2$ and Wool Production in units of $(100~\text{tonnes})^2$ with fixed-order models.}
    \centering
    \resizebox{\ifdim\width>\linewidth\linewidth\else\width\fi}{!}{
    \begin{tabular}[t]{l|lrrlrrrr}
    \toprule
     & Level & Training Base MSE & Test Base MSE & Recon. Type & Training Recon. MSE & Test Recon. MSE & Training rMSE & Test rMSE\\
    \midrule
    \multirow{12}{*}{\rotatebox[origin=c]{90}{A\&E Total Attendances}} & &  &  & Bottom-Up & 1219.43 & 1981.83 & 0.08 & \textbf{-0.09}\\
    
     & &  &  & Full Cov. & 1124.78 & 2222.40 & \textbf{0.00} & 0.03\\
    
     & &  &  & Spectral & 1169.64 & 2021.56 & 0.04 & -0.07\\
    
    & \multirow{-4}{*}{ Monthly} & \multirow{-4}{*}{ 1125.60} & \multirow{-4}{*}{ 2166.15} & OLS & 1132.26 & 2112.09 & 0.01 & -0.02\\
    \cmidrule{2-9}
     & &  &  & Bottom-Up & 150.19 & 170.12 & 0.00 & \textbf{0.00}\\
    
     & &  &  & Full Cov. & 148.26 & 186.73 & -0.01 & 0.10\\
    
     & &  &  & Spectral & 148.44 & 171.55 & -0.01 & 0.01\\
    
    & \multirow{-4}{*}{ Weekly} & \multirow{-4}{*}{ 150.19} & \multirow{-4}{*}{ 170.12} & OLS & 147.74 & 176.55 & \textbf{-0.02} & 0.04\\
    \cmidrule{2-9}
     & &  &  & Bottom-Up & 1369.62 & 2151.95 & 0.07 & \textbf{-0.08}\\
    
     & &  &  & Full Cov. & 1273.04 & 2409.12 & \textbf{0.00} & 0.03\\
    
     & &  &  & Spectral & 1318.08 & 2193.11 & 0.03 & -0.06\\
    
    & \multirow{-4}{*}{ Overall} & \multirow{-4}{*}{ 1275.79} & \multirow{-4}{*}{ 2336.26} & OLS & 1280.00 & 2288.64 & 0.00 & -0.02\\
    \midrule\midrule
    \multirow{16}{*}{\rotatebox[origin=c]{90}{Wool Production}}& &  &  & Bottom-Up & 156.23 & 293.50 & 0.05 & 1.23\\

    & &  &  & Full Cov. & 119.19 & 330.31 & \textbf{-0.20} & 1.51\\

    & &  &  & Spectral & 134.30 & 200.81 & -0.10 & 0.52\\

    & \multirow{-4}{*}{ Annual} & \multirow{-4}{*}{ 149.18} & \multirow{-4}{*}{ 131.83} & OLS & 141.12 & 146.46 & -0.05 & \textbf{0.11}\\
    \cmidrule{2-9}
    &  &  &  & Bottom-Up & 80.55 & 80.23 & 0.07 & 1.11\\

    &  &  &  & Full Cov. & 50.70 & 87.30 & \textbf{-0.33} & 1.30\\

    &  &  &  & Spectral & 59.84 & 53.83 & -0.21 & 0.42\\

    & \multirow{-4}{*}{ Biannual} & \multirow{-4}{*}{ 75.42} & \multirow{-4}{*}{ 37.98} & OLS & 66.77 & 40.65 & -0.11 & \textbf{0.07}\\
    \cmidrule{2-9}
    &  &  &  & Bottom-Up & 24.59 & 23.54 & 0.00 & 0.00\\

    &  &  &  & Full Cov. & 16.43 & 24.89 & \textbf{-0.33} & 0.06\\

    &  &  &  & Spectral & 19.11 & 16.60 & -0.22 & -0.29\\

    & \multirow{-4}{*}{ Quarterly} & \multirow{-4}{*}{ 24.59} & \multirow{-4}{*}{ 23.54} & OLS & 21.15 & 13.64 & -0.14 & \textbf{-0.42}\\
    \cmidrule{2-9}
    &  &  &  & Bottom-Up & 261.37 & 397.27 & 0.05 & 1.05\\

    &  &  &  & Full Cov. & 186.32 & 442.49 & \textbf{-0.25} & 1.29\\

    &  &  &  & Spectral & 213.25 & 271.25 & -0.14 & 0.40\\

    &  \multirow{-4}{*}{ Overall} & \multirow{-4}{*}{ 249.19} & \multirow{-4}{*}{ 193.35} & OLS & 229.04 & 200.75 & -0.08 & \textbf{0.04}\\
    \bottomrule
    \end{tabular}}
\end{table}

\subsection{Wool Production}
Another popular dataset is the \textit{woolyrnq} dataset, available from the \texttt{forecast} package in R (\citet{forecast-pkg}). It is about the quarterly production of woolen yarn in Australia, given in units of tonnes from March $1965$ to September $1994$. We aggregate the data to biannual as well as annual frequency yielding a $3$-level hierarchy with $k\in\{4,2,1\}$. In order to have complete observations we remove the partially observed last year $1994$. This then gives us $116$ quarters, $58$ half-years as well as $29$ years of data. As previously, we split the data into $80\%$ training data leading to $23$ training years.

In contrast to the A\&E data, we do not perform any preprocessing 
.A seasonality decomposition such as \texttt{stl} is not suitable for the annual time series, hence we do not perform it at all.

Table~\ref{tab:data:ta_wool2} presents the results for fixed-order models. According to AICc, the most suitable model for the quarterly time series is an $\text{ARIMA}(3,1,2)$ model, which is already quite complex. The theory of aggregated ARIMA models then gives us $\text{ARIMA}(3,1,3)$ and $\text{ARIMA}(3,1,4)$ models for the biannual and annual time series, respectively. Despite the relatively small sample sizes for the biannual and annual data, these high-complexity models do not seem to suffer from overfitting. Using automated model selection, the corresponding models would be $\text{ARIMA}(0,1,0)$ and $\text{ARIMA}(1,1,1)$, respectively, which produce very similar results. Therefore, we only present the results for the fixed-order case.

Nevertheless, we observe similar effects as with the A\&E data. The bottom-up approach performs worse on the training data compared to covariance-based reconciliation methods. On the test data, both the bottom-up approach and the full covariance method exhibit poor generalization, while the spectral and OLS methods perform better. Notably, the full covariance method generalizes even worse than the bottom-up approach, a consistent finding across all data examples and simulations.

\subsection{Additional Datasets}
We run experiments on a set of additional datasets and give an overall summary of the results. Based on the forecasting literature, especially hierarchical forecast reconciliation, we select the following $5$ datasets.
    \begin{itemize}
        \item Energy \citep{PANAGIOTELIS2023693}: Daily electricity generation per source, available from the author's GitHub repository\footnote{\url{https://github.com/PuwasalaG/Probabilistic-Forecast-Reconciliation}}.
        \item Food \citep{NEUBAUER2024}: Daily data from smart fridges with the goal of forecasting the demand for each fridge for the upcoming week in a one-step-ahead fashion.
        \item M3 \citep{M3Data}: Quarterly data of the M3 competition. The data was obtained from the R package Mcomp \citep{M3DataR}.
        \item Prison \citep{hyndman2018forecasting}: Quarterly data about Australian prison population per state.
        \item Tourism \citep{wick:opt_fc_recon, GiroEtAl2023}: Monthly data about visitor nights in Australian districts, taken from GitHub\footnote{\url{https://github.com/daniGiro/ctprob}}.
\end{itemize}

This selection of datasets covers a wide range of frequencies and domains, summarised in Table~\ref{tab:dat_props}. To ensure a non-singular covariance matrix estimate in order to be able to compute the full covariance reconciliation method, we maintain a relatively low order of aggregation. Specifically, we aggregate the energy data into weekly data, the M3 data into annual data, and so on. For each time series, we hold out $20\%$ of the data as test data. Table~\ref{tab:dat_res} also presents the training and test rMSE values for the selected reconciliation methods, summarized by trimmed means and corresponding standard errors. However, this presentation of the results does not provide much insight into the underlying dynamics. We observe that in-sample, the full covariance method performs well, but it does not generalize effectively. Similarly, the bottom-up approach does not produce the best results on the training data and also yields sub-optimal forecasts on the test data, contrary to the simulations. Comparing the two approaches we do observe that the full covariance method generalizes worse than the bottom-up method, confirming our simulation findings. Finally, the more sophisticated approach of utilizing the spectral decomposition performs well out-of-sample.

\begin{table}

\caption{\label{tab:dat_props}Dataset properties. $N$ denotes the number of total time series in the dataset, and $n_\text{top},n_\text{bottom}$ give the range of the available lengths in the hierarchy given by $k$.}
\centering
\begin{tabular}[t]{lllll}
\toprule
Dataset & $N$ & $n_{\text{top}}$ & $n_{\text{bottom}}$ & $k\in$\\
\midrule
Energy & 23 & 51-51 & 357-357 & $\{7,1\}$\\
Food & 122 & 7-107 & 35-535 & $\{5,1\}$\\
M3 & 756 & 8-18 & 32-72 & $\{4,1\}$\\
Prison & 8 & 12-12 & 48-48 & $\{4,1\}$\\
Tourism & 525 & 76-76 & 228-228 & $\{3,1\}$\\
\bottomrule
\end{tabular}
\end{table}

\begin{table}
\centering
\caption{\label{tab:dat_res}$10\%$-trimmed overall means for $5$ datasets and selected reconciliation methods. The standard errors are available in parentheses.}
\centering
\resizebox{\ifdim\width>\linewidth\linewidth\else\width\fi}{!}{
\begin{tabular}[t]{lrrrrrrrr}
\toprule
\multicolumn{1}{c}{ } & \multicolumn{4}{c}{Training rMSE} & \multicolumn{4}{c}{Test rMSE} \\
\cmidrule(l{3pt}r{3pt}){2-5} \cmidrule(l{3pt}r{3pt}){6-9}
Dataset & Bottom-Up & Full Cov. & OLS & Spectral & Bottom-Up & Full Cov. & OLS & Spectral\\
\midrule
Energy & -0.03 (0.02) & \textbf{-0.06 (0.01)} & -0.02 (0.00) & -0.06 (0.01) & -0.02 (0.05) & -0.02 (0.05) & -0.02 (0.00) & \textbf{-0.04 (0.03)}\\
Food & 0.04 (0.01) & \textbf{-0.03 (0.01)} & -0.01 (0.00) & -0.01 (0.00) & 0.04 (0.02) & 0.01 (0.01) & -0.01 (0.00) & \textbf{0.00 (0.01)}\\
M3 & -0.17 (0.02) & \textbf{-0.28 (0.01)} & -0.11 (0.00) & -0.27 (0.01) & -0.09 (0.03) & -0.13 (0.03) & -0.11 (0.01) & \textbf{-0.19 (0.02)}\\
Prison & \textbf{-0.18 (0.12)} & -0.12 (0.17) & -0.11 (0.02) & 0.00 (0.18) & \textbf{-0.40 (0.12)} & -0.30 (0.14) & -0.14 (0.03) & -0.01 (0.20)\\
Tourism & 0.03 (0.00) & \textbf{-0.05 (0.00)} & 0.00 (0.00) & -0.01 (0.00) & 0.01 (0.01) & \textbf{-0.02 (0.01)} & -0.01 (0.00) & -0.01 (0.00)\\
\bottomrule
\end{tabular}}
\end{table}

\begin{figure}[!ht]
    \includegraphics[width=\textwidth]{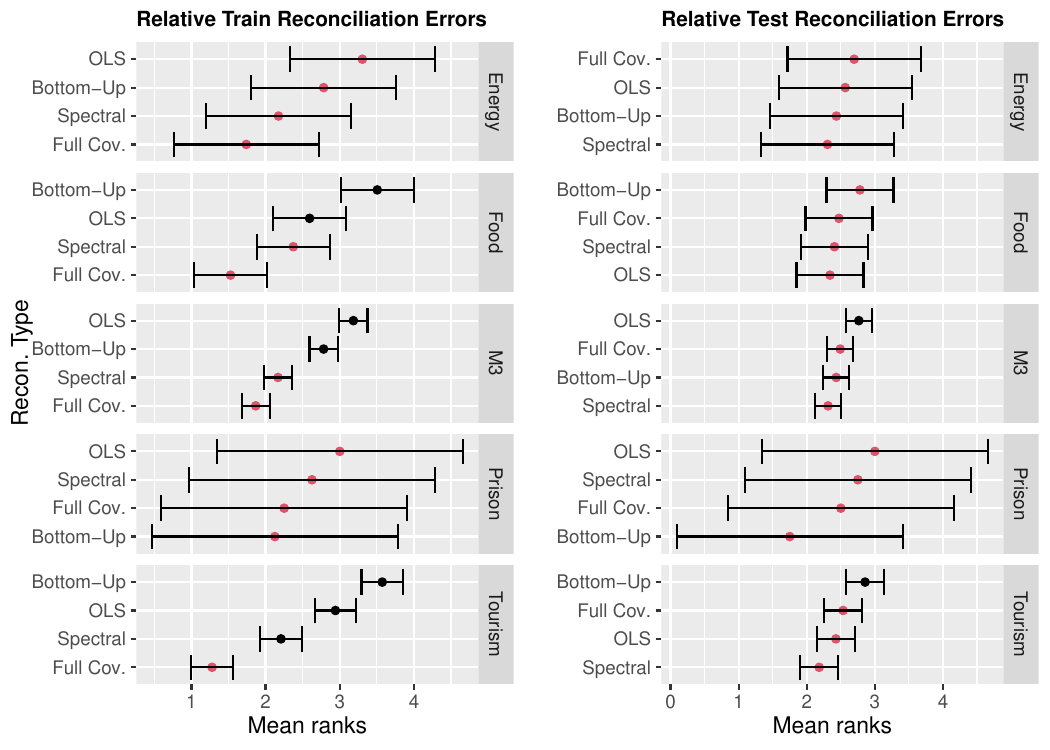}
    \caption{MCB test with confidence level $0.95$ on the overall level.}\label{fig:dat_all_nemenyi}
\end{figure}

We conduct an accuracy ranking based on multiple comparisons with the best (MCB) test, introduced by \citet{KONING2005397}, for each dataset, divided into training and test data. Figure~\ref{fig:dat_all_nemenyi} clearly demonstrates the statistically superior performance of the full covariance method compared to the bottom-up approach in-sample, while the performance difference becomes practically negligible on the test data, consistent with our theory and simulations.

Additionally, Figure~\ref{fig:dat_all_percs} presents percentile plots comparing the four different approaches. These plots further illustrate that while the full covariance method performs well in-sample, its performance significantly deteriorates out-of-sample. Specifically, on the training data, more forecasts are improved by full covariance reconciliation, but this relationship largely reverses on the test data.

\begin{figure}[!ht]
    \includegraphics[width=\textwidth]{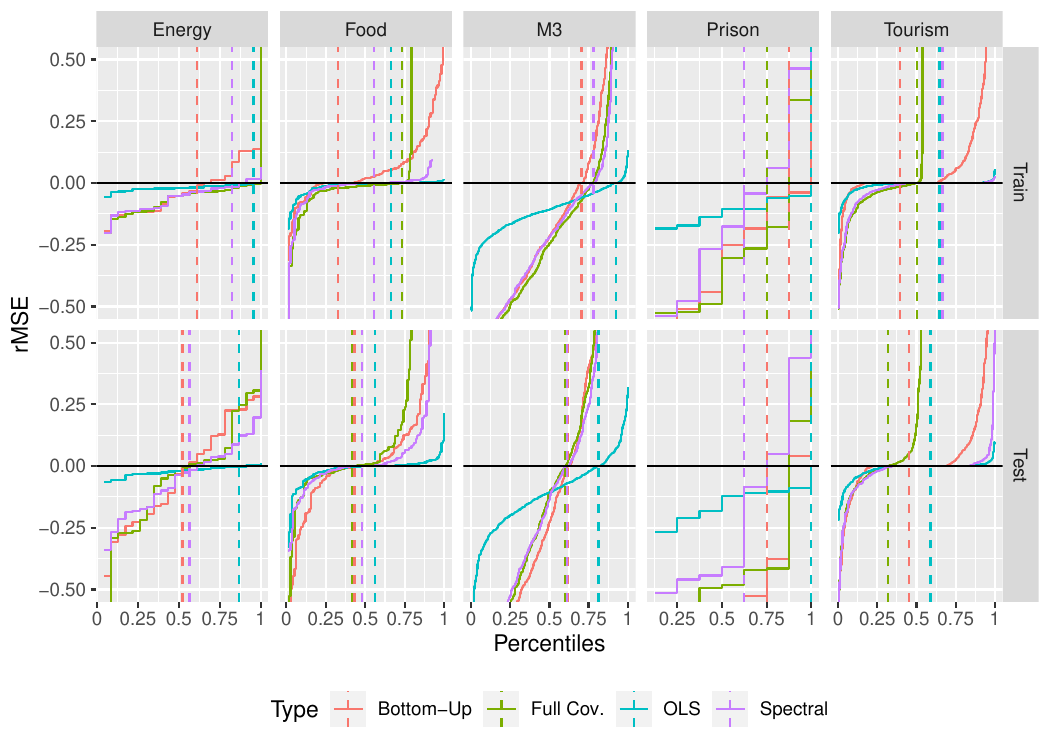}
    \caption{Percentile plots for each dataset, split by training and test set.}\label{fig:dat_all_percs}
\end{figure}

\section{Conclusions}\label{sec:concl}
In this paper, we explored the theoretical implications of applying the minimum trace reconciliation method within the context of temporal hierarchies. By examining temporally aggregated ARMA models, we demonstrated that the optimal reconciliation method, when based on the true covariance matrix, is equivalent to a bottom-up approach. Our extensive simulation studies tested this theory across various scenarios involving different model complexities, hierarchy structures, and levels of uncertainty. The findings support our theory, indicating that the bottom-up method is a viable approach. This aligns with numerous literature findings where the bottom-up approach consistently produces useful results in suitable settings.

The simulation results also reveal that in-sample, covariance-based minimum trace reconciliation methods outperform the simple bottom-up approach. However, this relationship reverses out-of-sample, with the bottom-up approach generalizing better on the test data compared to the full covariance matrix across simulations and data examples. Further research is necessary to understand why this effect occurs so markedly. In misspecified scenarios similar conclusions can be drawn. When having underfitting base models on a lower level, the bottom-up approach should be used with caution while in remaining scenarios sophisticated reconciliation approaches perform similarly. Additionally, other estimators were tested and showed improved performance over using the full covariance matrix in certain settings, highlighting the potential for the ongoing research of new temporal hierarchical covariance estimators in the minimum trace approach.

Overall, our work contributes to the field of temporal forecast reconciliation by linking it to temporally aggregated ARMA models. We have theoretically established that the bottom-up approach is the optimal reconciliation method and reinforced this with comprehensive simulation studies and data illustrations. This supports the use of the bottom-up method in both theoretical and practical applications.

\section*{Computational details}
The simulations and data examples were carried out in \textbf{R} 4.3.0. The corresponding source code of this paper in the form of an R package is available from GitHub at \url{https://github.com/neubluk/FTATS}. For convenience, all datasets except the $\text{M}3$ dataset are included in the package.


\section*{Acknowledgments and Disclosure of Funding}
We acknowledge support from the Austrian Research Promotion Agency (FFG), Basisprogramm project “Meal Demand Forecast” and Schrankerl GmbH for the cooperation and access to their data. We further acknowledge funding from the Austrian Science
Fund (FWF) for the project “High-dimensional statistical learning: New methods to
advance economic and sustainability policies” (ZK 35), jointly carried out by WU
Vienna University of Economics and Business, Paris Lodron University Salzburg, TU
Wien, and the Austrian Institute of Economic Research (WIFO). 

\newpage

\appendix
\section{Calculations and Proofs}\label{app:calc_proofs}

As in \citet{Silvestrini2005TEMPORALAO}, we illustrate this framework based on an $\text{AR}(1)$ model. Let $y_t\sim\text{AR}(1)$ be centered at $0$ with AR parameter $\phi\in(-1,1)$ and innovation variance $\sigma^2$. According to Eq.~\eqref{eq:agg_arima} we obtain $y_T^\ast\sim\text{ARMA}(1,1)$ for any $k>1$ and AR parameter $\beta = \phi^k$. The MA parameter $\eta$ as well as the noise $\sigma_{\ast}^2$ are computed as follows.

For lags $0,1$ we compute the autocovariances of $(1+\eta B)\epsilon^\ast_T$ with $B=L^k$ and $T(L)\epsilon_t$ with the aggregation polynomial $T(L)$ given by
\begin{align}
    T(L) &= \frac{1-\delta^k L^k}{1-\delta L}\frac{1-L^k}{1-L} \\
        &= \sum_{i=0}^{k-1} \delta^i L^i \sum_{j=0}^{k-1} L^j,
\end{align}
with $\delta = \phi^{-1}$ being the inverse root of the corresponding AR polynomial and $L$ being the lag operator such that $Ly_t = Ly_{t-1}$.

Because the MA order is $1$, all lags greater than $1$ are zero. First note that
\begin{align*}
    T(L)\epsilon_t = (1,\phi,\dots,\phi^{k-1})
    \overbrace{\begin{pmatrix}
        1 & \dots &\dots & \dots & 1 & 0 & \dots & \dots & 0 \\
        0 & 1 & \dots & \dots & \vdots & 1 & 0 & \dots & 0 \\
        \vdots & \ddots & \ddots & \ddots & \vdots & \vdots & \ddots & \ddots & \vdots \\
        \vdots & \ddots & \ddots & \ddots & \vdots & \vdots & \ddots & \ddots & \vdots \\
        \undermat{k\times k}{0 & \dots & \dots & 0 & 1 &} \undermat{k\times (k-1)}{1 & \dots & \dots & 1}
    \end{pmatrix}}^{=A}
    \begin{pmatrix}
        \epsilon_t\\
        \vdots\\
        \epsilon_{t-(2k-2)}
    \end{pmatrix}.
\end{align*}

Next, we set up the equations based on the auto-correlation functions to determine $\eta$ and $\sigma_{\ast}^2$.

To this end, the variances are computed to be
    \begin{align}\label{eq:agg_ar1_var}
        \gamma^\ast(0) &= \text{Var}( (1+\eta B)\epsilon^\ast_T) \nonumber\\
        &= (1+\eta^2)\sigma_{\ast}^2,
    \end{align}
which must be equal to 
    \begin{align}\label{eq:agg_ar1_var2}
        \gamma(0) &= \text{Var}(T(L)\epsilon_t) \nonumber \\
        &=\sigma^2 (1,\phi,\dots,\phi^{k-1})AA'(1,\phi,\dots,\phi^{k-1})' \nonumber \\
        &= \sigma^2\left(\sum_{j=0}^{k-1} \left(\sum_{i=0}^j \phi^i\right)^2 + \sum_{j=0}^{k-1} \left(\sum_{i=j}^{k-1} \phi^i\right)^2\right).
    \end{align}
Similarly, the lag $1$ auto-covariances are
    \begin{align}\label{eq:agg_ar1_acf1}
        \gamma^\ast(1) &= \text{Cov}( (1+\eta B)\epsilon^\ast_T, (1+\eta B)\epsilon^\ast_{T-1}) \nonumber\\
        &= \eta\sigma_{\ast}^2,
    \end{align}
with needed equality to
    \begin{align}\label{eq:agg_ar1_acf1_2}
        \gamma(1) &= \text{Cov}( T(L)\epsilon_t, T(L)\epsilon_{t-k}) \nonumber \\ 
        &=\sigma^2 (1,\phi,\dots,\phi^{k-1})ACA'(1,\phi,\dots,\phi^{k-1}) \nonumber \\
        &= \sigma^2\left(\sum_{j=1}^{k-1} \left(\sum_{i=j}^{k-1}\phi^i\sum_{l=0}^{j-1}\phi^l\right)\right)
    \end{align}
    where 
    \begin{align*}
        C&=\frac{1}{\sigma^2}\text{Cov}\left( (\epsilon_t \dots \epsilon_{t-(2k-2)})',
            (\epsilon_{t-k}, \dots, \epsilon_{t-k-(2k-2)})'
            \right)\\
        &=\begin{pmatrix}
            0_{k\times (k-1)} & 0_{k\times k} \\
            I_{k-1} & 0_{(k-1)\times k}
        \end{pmatrix}
    \end{align*}

Solving the system of equations $\gamma(0)=\gamma^\ast(0),\gamma(1)=\gamma^\ast(1)$ using \eqref{eq:agg_ar1_var}-\eqref{eq:agg_ar1_acf1_2} yields
\begin{align*}
    \sigma_{\ast}^2 &= \sigma^2 \frac{(1,\phi,\dots,\phi^{k-1})AA'(1,\phi,\dots,\phi^{k-1})'}{1+\eta^2}\\
    \eta &= (1+\eta^2)\rho_1,
\end{align*}

where $\rho_1=\frac{\gamma(1)}{\gamma(0)}=\frac{\gamma^\ast(1)}{\gamma^\ast(0)}$ denotes the auto-correlation value at lag $1$.
\begin{proof}[Proof of Lemma~\ref{lm:agg_cov}]
    First, we compute the $h$-step forecasts of the disaggregated series for $h=1,\dots,k$. For the $\text{AR}(1)$ process this can be done recursively and we obtain residuals given by
    \begin{align}
        e_t^{(h)} = \sum_{i=0}^{h-1}\phi^i \epsilon_{t+h-i}.
    \end{align}   
    The corresponding pairwise covariances are quickly computed for $h_1\leq h_2$ by
    \begin{align}
        \text{Cov}\left(e_t^{(h_1)},e_t^{(h_2)}\right) &= \sigma^2 \sum_{l=0}^{h_1-1}\phi^{h_2-h_1+2l}\\
        &= \sigma^2 \phi^{h_2-h_1}\frac{1-\phi^{2h_1}}{1-\phi^2}, \\
    \end{align}
    hence for $\mathbf e_t = \left(e_t^{(1)},\dots,e_t^{(k)}\right)'$ we obtain the covariance matrix on the bottom level
    $\text{Cov}(\mathbf e_t) = \sigma^2 \Phi\Phi'$.
    
    For $y^\ast_T$ we perform a $1$-step forecast, thus ${e^\ast_T}^{(1)}=\epsilon^\ast_{T+1}$ with $\text{Var}({e^\ast_T}^{(1)})=\sigma_{\ast}^2$. To compute $\text{Cov}({e^\ast_T}^{(1)},e_t^{(h)})$, we do as follows. First, write $\epsilon^\ast_{T+1}=y^\ast_{T+1}-\beta y^\ast_T-\eta\epsilon^\ast_T$, then for $T=tk$ and $j=1,\dots,k$ we have 
    \begin{align}
        \text{Cov}(\epsilon^\ast_{T+1}, \epsilon_{tk+j}) &= \sum_{i=0}^{k-1} \text{Cov}(y_{tk+k-i}, \epsilon_{tk+j}) \\
        &= \sum_{i=0}^{k-1} \sum_{l=0}^{tk+k-i}\phi^l \text{Cov}(\epsilon_{tk+k-i-l}, \epsilon_{tk+j}) \\
        &= \sigma^2\sum_{i=0}^{k-j}\phi^i \\
        &= \sigma^2 \frac{1-\phi^{k-j+1}}{1-\phi},
    \end{align} 
    since $\text{Cov}(\epsilon_{tk+k-i-l}, \epsilon_{tk+j})=\sigma^2$ if $l=k-i-j$ and $0$ otherwise. Together, we obtain the temporal cross-covariances of
    \begin{align}
        \text{Cov}({e^\ast_T}^{(1)},e_{tk}^{(h)}) &= \text{Cov}({e^\ast_T}^{(1)},\sum_{i=0}^{h-1}\phi^i\epsilon_{tk+h-i}) \\
        &= \frac{\sigma^2}{1-\phi}\left( \frac{1-\phi^h}{1-\phi} -  \phi^{k-h+1}\frac{1-\phi^{2h}}{1-\phi^2}\right),
    \end{align}  
    hence the cross-covariance vector is given by
    \begin{align}
        \text{Cov}(e^\ast_T,\mathbf e_{tk}) &= \sigma^2(1,\dots,1) \tilde\Phi\tilde\Phi.
    \end{align}   
\end{proof}

\begin{proof}[Proof of Theorem~\ref{thm:mint_bu}]
    The minimizer of Eq.~\eqref{eq:mint} is given by $G^\ast=(S'W_1^{-1}S)^{-1}S'W_1^{-1}$. First, note that 
    \begin{align}
        W_1^{-1}S = \begin{pmatrix}
            \mathbf 0_k'\\
            (\sigma^2\Phi\Phi')^{-1}
        \end{pmatrix},
    \end{align}
    due to $\text{Cov}(e^\ast_T,\mathbf e_{tk}) = \sigma^2\mathbf 1_k' \Phi\Phi'$. Then the minimizing $G^\ast$ matrix is obtained to be $G^\ast=(\mathbf 0_k~\mathbf I_k)$ and hence
    \begin{align}
        SG^\ast = \begin{pmatrix}
            0 & \mathbf 1_k' \\
            \mathbf 0_k & \mathbf I_k
        \end{pmatrix},
    \end{align}
    which is exactly the bottom-up forecast for the aggregated series.
\end{proof}

\section{Additional Results}
\begin{table}
\centering
\caption{\label{tab:h=1, k=(1,4), sigma.sq=1, auto=TRUE}Mean rMSE per buckets of $\phi$ for $ h=1, k\in\{4,1\}, \sigma^2=1 $ and auto-selected models. The standard errors are given in parenthesis.}
\centering
\resizebox{\ifdim\width>\linewidth\linewidth\else\width\fi}{!}{
\begin{tabular}[t]{lllrrrrrr}
\toprule
\multicolumn{3}{c}{ } & \multicolumn{3}{c}{Training rMSE} & \multicolumn{3}{c}{Test rMSE} \\
\cmidrule(l{3pt}r{3pt}){4-6} \cmidrule(l{3pt}r{3pt}){7-9}
Level & n & Recon. Type & {}[-0.9,-0.5] & (-0.5,0.5] & (0.5,0.9] & {}[-0.9,-0.5] & (-0.5,0.5] & (0.5,0.9]\\
\midrule
 &  & Bottom-Up & -0.02 (0.01) & 0.03 (0.01) & -0.08 (0.02) & \textbf{-0.05 (0.03)} & -0.01 (0.01) & -0.05 (0.03)\\

 &  & Full Cov. & \textbf{-0.10 (0.01)} & - & \textbf{-0.17 (0.01)} & 0.11 (0.03) & - & -0.01 (0.04)\\

 &  & Spectral & -0.09 (0.01) & \textbf{-0.03 (0.00)} & -0.14 (0.01) & 0.01 (0.02) & 0.01 (0.01) & -0.01 (0.05)\\

 & \multirow{-4}{*}{ 20} & OLS & -0.03 (0.00) & -0.01 (0.00) & -0.06 (0.00) & -0.04 (0.00) & \textbf{-0.02 (0.00)} & \textbf{-0.07 (0.01)}\\
\cmidrule{2-9}
 &  & Bottom-Up & -0.04 (0.01) & 0.01 (0.00) & -0.09 (0.01) & \textbf{-0.08 (0.01)} & \textbf{-0.04 (0.01)} & -0.12 (0.02)\\

 &  & Full Cov. & \textbf{-0.08 (0.01)} & - & \textbf{-0.14 (0.01)} & -0.04 (0.01) & - & \textbf{-0.12 (0.01)}\\

 &  & Spectral & -0.06 (0.00) & \textbf{-0.02 (0.00)} & -0.12 (0.01) & -0.04 (0.01) & -0.02 (0.01) & -0.11 (0.02)\\

 & \multirow{-4}{*}{ 50} & OLS & -0.03 (0.00) & -0.01 (0.00) & -0.05 (0.00) & -0.03 (0.00) & -0.02 (0.00) & -0.06 (0.00)\\
\cmidrule{2-9}
 &  & Bottom-Up & -0.05 (0.00) & 0.00 (0.00) & -0.10 (0.01) & \textbf{-0.08 (0.01)} & \textbf{-0.03 (0.00)} & -0.13 (0.01)\\

 &  & Full Cov. & \textbf{-0.07 (0.00)} & - & \textbf{-0.13 (0.00)} & -0.06 (0.01) & - & \textbf{-0.14 (0.01)}\\

 &  & Spectral & -0.06 (0.00) & \textbf{-0.01 (0.00)} & -0.12 (0.00) & -0.05 (0.01) & -0.02 (0.00) & -0.13 (0.01)\\

\multirow{-12}{*}[1\dimexpr\aboverulesep+\belowrulesep+\cmidrulewidth]{ Level 1} & \multirow{-4}{*}{ 100} & OLS & -0.03 (0.00) & -0.01 (0.00) & -0.05 (0.00) & -0.03 (0.00) & -0.01 (0.00) & -0.06 (0.00)\\
\cmidrule{1-9}
 &  & Bottom-Up & 0.00 (0.00) & 0.00 (0.00) & 0.00 (0.00) & \textbf{ 0.00 (0.00)} & \textbf{ 0.00 (0.00)} & \textbf{ 0.00 (0.00)}\\

 &  & Full Cov. & -0.02 (0.01) & - & \textbf{-0.07 (0.01)} & 0.12 (0.02) & - & 0.08 (0.02)\\

 &  & Spectral & \textbf{-0.02 (0.00)} & \textbf{-0.02 (0.00)} & -0.03 (0.01) & 0.04 (0.01) & 0.03 (0.01) & 0.05 (0.01)\\

 & \multirow{-4}{*}{ 20} & OLS & 0.00 (0.00) & -0.01 (0.00) & 0.05 (0.01) & 0.01 (0.00) & 0.03 (0.01) & 0.12 (0.03)\\
\cmidrule{2-9}
 &  & Bottom-Up & 0.00 (0.00) & 0.00 (0.00) & 0.00 (0.00) & \textbf{ 0.00 (0.00)} & \textbf{ 0.00 (0.00)} & \textbf{ 0.00 (0.00)}\\

 &  & Full Cov. & \textbf{-0.02 (0.00)} & - & \textbf{-0.04 (0.00)} & 0.03 (0.00) & - & 0.02 (0.01)\\

 &  & Spectral & -0.01 (0.00) & \textbf{-0.01 (0.00)} & -0.02 (0.00) & 0.01 (0.00) & 0.02 (0.00) & 0.02 (0.01)\\

 & \multirow{-4}{*}{ 50} & OLS & 0.00 (0.00) & 0.00 (0.00) & 0.04 (0.01) & 0.01 (0.00) & 0.02 (0.00) & 0.09 (0.01)\\
\cmidrule{2-9}
 &  & Bottom-Up & 0.00 (0.00) & 0.00 (0.00) & 0.00 (0.00) & \textbf{ 0.00 (0.00)} & \textbf{ 0.00 (0.00)} & \textbf{ 0.00 (0.00)}\\

 &  & Full Cov. & \textbf{-0.01 (0.00)} & - & \textbf{-0.02 (0.00)} & 0.01 (0.00) & - & 0.01 (0.01)\\

 &  & Spectral & 0.00 (0.00) & \textbf{-0.01 (0.00)} & -0.01 (0.00) & 0.01 (0.00) & 0.01 (0.00) & 0.01 (0.00)\\

\multirow{-12}{*}[1\dimexpr\aboverulesep+\belowrulesep+\cmidrulewidth]{ Level 2} & \multirow{-4}{*}{ 100} & OLS & 0.00 (0.00) & 0.00 (0.00) & 0.04 (0.00) & 0.01 (0.00) & 0.01 (0.00) & 0.08 (0.01)\\
\cmidrule{1-9}
 &  & Bottom-Up & -0.02 (0.01) & 0.02 (0.01) & -0.08 (0.01) & \textbf{-0.04 (0.01)} & \textbf{-0.02 (0.01)} & \textbf{-0.07 (0.03)}\\

 &  & Full Cov. & \textbf{-0.06 (0.01)} & - & \textbf{-0.16 (0.01)} & 0.09 (0.02) & - & -0.03 (0.03)\\

 &  & Spectral & -0.06 (0.00) & \textbf{-0.03 (0.00)} & -0.14 (0.01) & 0.01 (0.01) & 0.00 (0.01) & -0.03 (0.04)\\

 & \multirow{-4}{*}{ 20} & OLS & -0.02 (0.00) & -0.01 (0.00) & -0.06 (0.00) & -0.02 (0.00) & -0.02 (0.00) & -0.06 (0.00)\\
\cmidrule{2-9}
 &  & Bottom-Up & -0.02 (0.00) & 0.01 (0.00) & -0.08 (0.01) & \textbf{-0.05 (0.01)} & \textbf{-0.04 (0.01)} & -0.11 (0.02)\\

 &  & Full Cov. & \textbf{-0.05 (0.00)} & - & \textbf{-0.13 (0.01)} & -0.02 (0.01) & - & \textbf{-0.12 (0.01)}\\

 &  & Spectral & -0.04 (0.00) & \textbf{-0.01 (0.00)} & -0.12 (0.01) & -0.02 (0.00) & -0.02 (0.01) & -0.11 (0.01)\\

 & \multirow{-4}{*}{ 50} & OLS & -0.01 (0.00) & -0.01 (0.00) & -0.05 (0.00) & -0.02 (0.00) & -0.01 (0.00) & -0.06 (0.00)\\
\cmidrule{2-9}
 &  & Bottom-Up & -0.03 (0.00) & 0.00 (0.00) & -0.09 (0.01) & \textbf{-0.05 (0.00)} & \textbf{-0.02 (0.00)} & -0.13 (0.01)\\

 &  & Full Cov. & \textbf{-0.04 (0.00)} & - & \textbf{-0.12 (0.00)} & -0.03 (0.00) & - & \textbf{-0.13 (0.01)}\\

 &  & Spectral & -0.03 (0.00) & \textbf{-0.01 (0.00)} & -0.11 (0.00) & -0.03 (0.00) & -0.01 (0.00) & -0.12 (0.01)\\

\multirow{-12}{*}[1\dimexpr\aboverulesep+\belowrulesep+\cmidrulewidth]{ Overall} & \multirow{-4}{*}{ 100} & OLS & -0.01 (0.00) & 0.00 (0.00) & -0.04 (0.00) & -0.02 (0.00) & -0.01 (0.00) & -0.05 (0.00)\\
\bottomrule
\end{tabular}}
\end{table}

\section{Additional Plots}\label{app:plots}

\begin{figure}[!ht]
    \includegraphics[width=\textwidth]{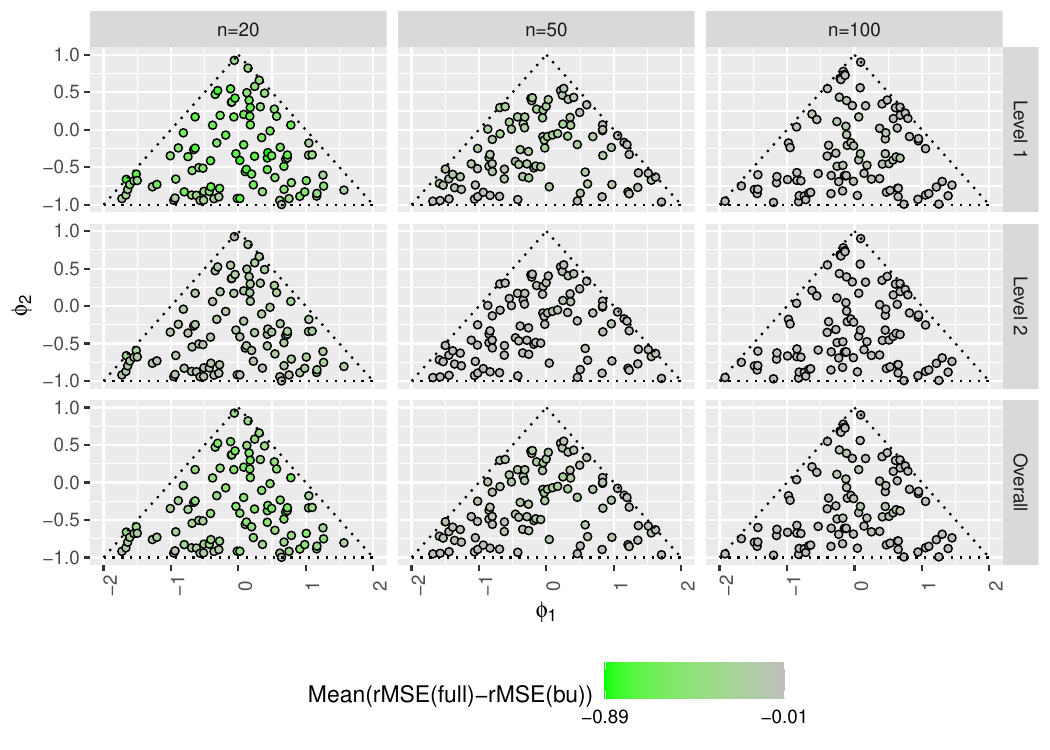}
    \caption{In-sample mean rMSE differences of the full covariance matrix and bottom-up reconciliation for $h=1,k\in\{4,1\},\sigma^2=1$ and fixed-order models.}
    \label{fig:ar2_1_14_1_train}
\end{figure}

\begin{figure}[!ht]
    \includegraphics[width=\textwidth]{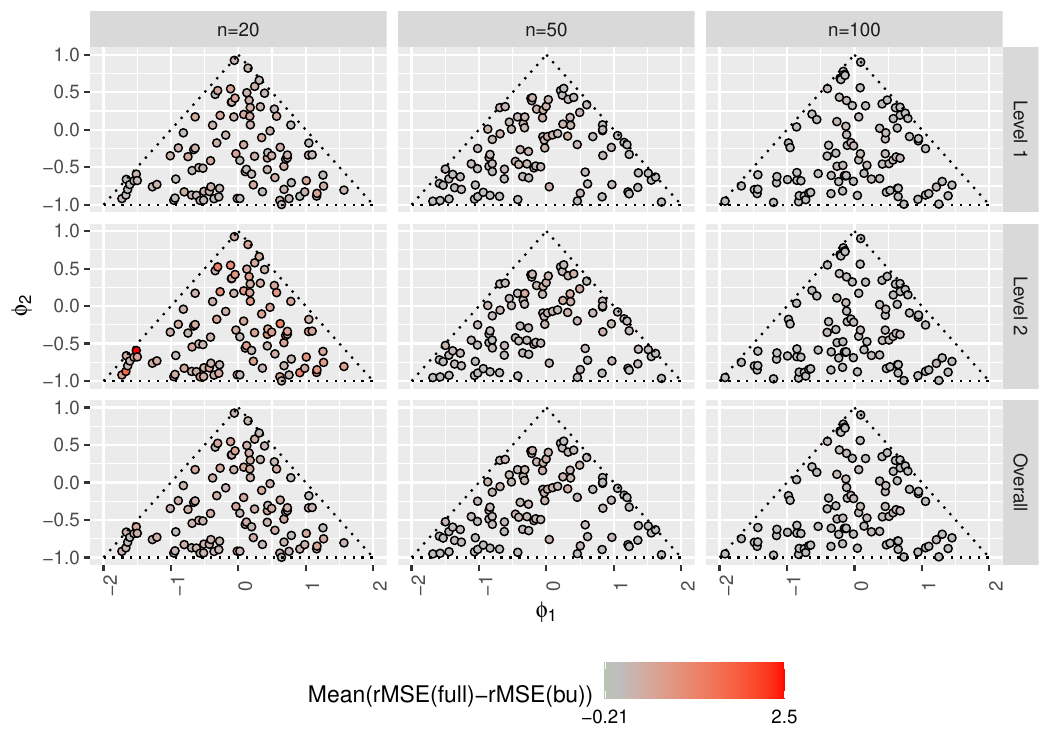}
    \caption{Out-of-sample mean rMSE differences of the full covariance matrix and bottom-up reconciliation for $h=1,k=(1,4),\sigma^2=1$ and fixed-order models.}
    \label{fig:ar2_1_14_1_test}
\end{figure}

\begin{figure}[!ht]
    \includegraphics[width=\textwidth]{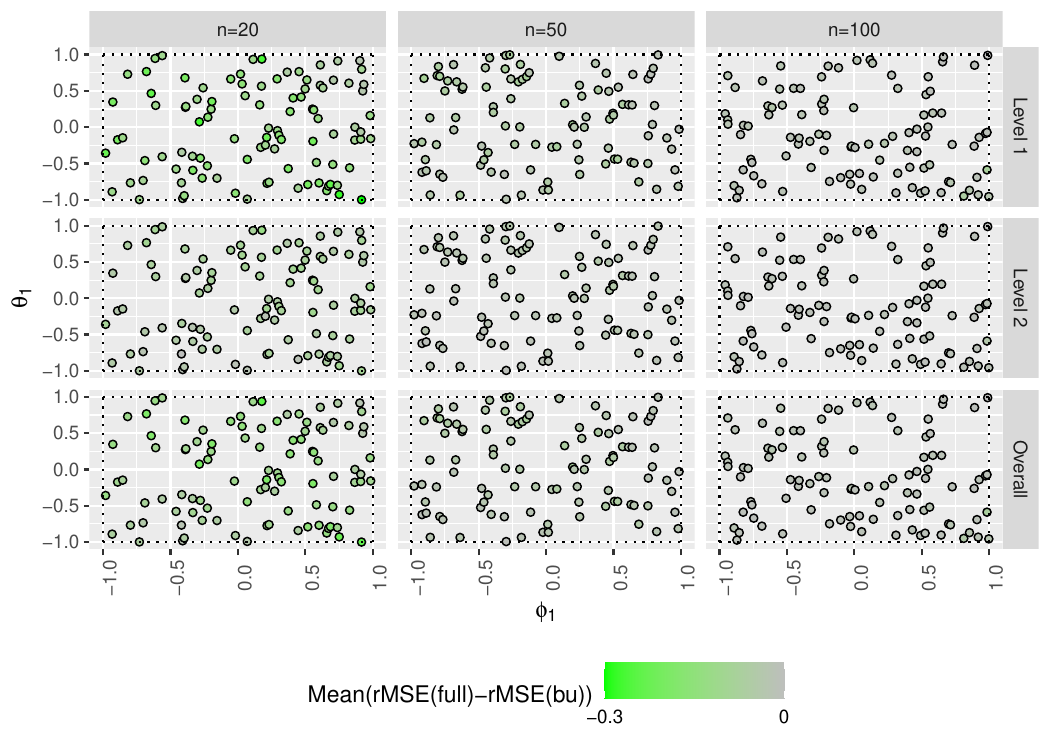}
    \caption{In-sample mean rMSE differences of the full covariance matrix and bottom-up reconciliation for $h=1,k\in\{4,1\},\sigma^2=1$ and fixed-order models.}
    \label{fig:arma11_1_14_1_train}
\end{figure}

\begin{figure}[!ht]
    \includegraphics[width=\textwidth]{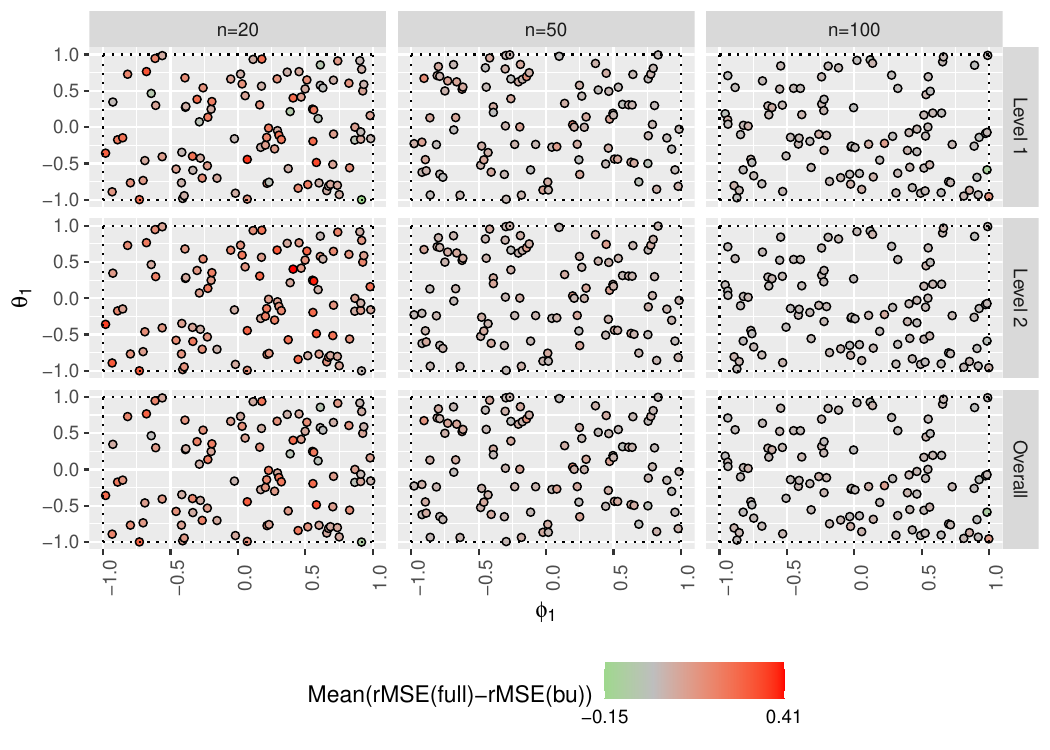}
    \caption{Out-of-sample mean rMSE differences of the full covariance matrix and bottom-up reconciliation for $h=1,k=(1,4),\sigma^2=1$ and fixed-order models.}
    \label{fig:arma11_1_14_1_test}
\end{figure}

\bibliographystyle{apalike}
\bibliography{refs}
\end{document}